\documentclass[11pt]{scrartcl}

\usepackage{url}
\usepackage[hidelinks]{hyperref}
\usepackage[utf8]{inputenc}
\usepackage[small]{caption}
\usepackage{graphicx}
\usepackage{amsmath}
\usepackage{booktabs}
\usepackage{amsthm}
\usepackage{amssymb}

\usepackage{cleveref}
\usepackage{crossreftools} 
\usepackage{color}
\usepackage{xspace}
\usepackage{tikz}
\usepackage{tabularx}
\usepackage{pbox}
\usepackage{framed}
\usepackage{enumitem}
\usepackage{arydshln}

\usepackage{algorithm,algpseudocode}
\usepackage{mathtools}
\usepackage[normalem]{ulem} 
\usepackage{comment}

\DeclareOldFontCommand{\bf}{\normalfont\bfseries}{\mathbf}

\usepackage{pifont,xcolor}
\newcommand{\cmark}{\ding{51}}
\newcommand{\xmark}{\ding{55}}
\newcommand{\yes}{\textcolor{green!50!black}{\cmark}}
\newcommand{\no}{\textcolor{red!50!black}{\xmark}}

\newcommand*{\diff}[1]{\mathop{}\!\mathrm{d}#1} 
\DeclarePairedDelimiter\ceiling{\lceil}{\rceil} 
\DeclarePairedDelimiter\floor{\lfloor}{\rfloor} 

\newtheorem{theorem}{Theorem}[section]
\newtheorem{corollary}[theorem]{Corollary}
\newtheorem{proposition}[theorem]{Proposition}
\newtheorem{lemma}[theorem]{Lemma}

\theoremstyle{definition}
\newtheorem{definition}[theorem]{Definition}
\newtheorem{example}[theorem]{Example}

\newtheorem{claim}{Claim}

\usepackage{natbib}

\makeatletter 
\newcommand{\EJR}[1]{\if\relax\detokenize\expandafter{\@firstofone#1{}}\relax EJR\else EJR-#1\fi\xspace}
\makeatother
\newcommand{\EJROne}{\EJR{1}}
\newcommand{\EJRM}{\EJR{M}}
\newcommand{\strongEJROne}{strong \EJROne}
\newcommand{\size}{\alpha}
\newcommand{\rulex}{MES\xspace}
\newcommand{\RuleX}{Method of Equal Shares\xspace}
\newcommand{\generalizedMESAbbr}{Generalized \rulex}
\newcommand{\generalizedMES}{Generalized \RuleX}

\allowdisplaybreaks

\begin{document}

\title{Approval-Based Voting with Mixed Goods}

\author{
Xinhang Lu\\ UNSW Sydney
\and
Jannik Peters\\ TU Berlin
\and
Haris Aziz\\ UNSW Sydney
\and
Xiaohui Bei\\NTU
\and
Warut Suksompong\\NUS
}

\date{\vspace{-3ex}}

\maketitle

\begin{abstract}
We consider a voting scenario in which the resource to be voted upon may consist of both indivisible and divisible goods.
This setting generalizes both the well-studied model of multiwinner voting and the recently introduced model of cake sharing.
Under approval votes, we propose two variants of the extended justified representation (EJR) notion from multiwinner voting, a stronger one called \emph{EJR for mixed goods (\mbox{EJR-M})} and a weaker one called \emph{EJR up to $1$ (\mbox{EJR-1})}.
We extend three multiwinner voting rules to our setting---GreedyEJR, the method of equal shares (MES), and proportional approval voting (PAV)---and show that while all three generalizations satisfy EJR-1, only the first one provides \mbox{EJR-M}.
In addition, we derive tight bounds on the proportionality degree implied by EJR-M and EJR-1, and investigate the proportionality degree of our proposed rules.
\end{abstract}

\section{Introduction}

In \emph{multiwinner voting}---a ``new challenge for social choice theory'', as \citet{FaliszewskiSkSl17} put it---the goal is to select a subset of candidates of fixed size from a given set based on the voters' preferences.
The candidates could be politicians vying for seats in the parliament, products to be shown on a company website, or places to visit on a school trip.
A common way to elicit preferences from the voters is via the \emph{approval} model, wherein each voter simply specifies the subset of candidates that he or she approves \citep{Kilgour10,LacknerSk22}.
While (approval-based) multiwinner voting has received substantial attention from (computational) social choice researchers in the past few years, a divisible analog called \emph{cake sharing} was recently introduced by \citet{BeiLuSu24}.
In cake sharing, the candidates correspond to a divisible resource such as time periods for using a facility or files to be stored in cache memory.
Following the famous resource allocation problem of \emph{cake cutting} \citep{RobertsonWe98,Procaccia16}, this divisible resource is referred to as a ``cake'', and cake sharing is the collective choice problem of selecting a subset of this resource.

In this paper, we study a setting that simultaneously generalizes both multiwinner voting and cake sharing, which we call \emph{(approval-based) voting with mixed goods}.
Specifically, in our setting, the resource may consist of both indivisible and divisible goods.\footnote{Since a ``candidate'' usually refers to an indivisible entity, we use the term ``good'' instead from here on.}
This generality allows our model to capture more scenarios than either of the previous models.
For example, when reserving time slots, it is possible that some hourly slots must be reserved as a whole, while other slots can be booked fractionally.
Likewise, in cache memory storage, certain files may need to be stored in their entirety, whereas other files can be broken into smaller portions.
Combinations of divisible and indivisible goods have been examined in the context of \emph{fair division}, where the resource is to be divided among interested agents and the entire resource can be allocated \citep{BeiLiLi21,BeiLiLu21,BhaskarSrVa21,KawaseNiSu23,NishimuraSu23}.
By contrast, we investigate mixed goods in a \emph{collective choice} context, where only a subset of the resource can be allocated but the allocated resource is collectively shared by all agents.\footnote{We henceforth use the term ``agent'' instead of ``voter''.}

There are multiple criteria that one can use to select a collective subset of resource based on the approval votes.
For example, one could try to optimize the \emph{social welfare}---the sum of the agents' utilities---or the \emph{coverage}---the number of agents who receive nonzero utility.
A representation criterion that has attracted growing interest is \emph{justified representation (JR)} \citep{AzizBrCo17}.
In multiwinner voting, if there are $n$~agents and $k$ (indivisible) goods can be chosen, then JR requires that whenever a group of at least $n/k$ agents approve a common good, some agent in that group must have an approved good in the selected set.
A well-studied strengthening of JR is \emph{extended justified representation (EJR)}, which says that for each positive integer $t$, if a group of at least $t\cdot n/k$ agents approve no fewer than $t$ common goods (such a group is said to be \emph{$t$-cohesive}), some agent in that group must have no fewer than $t$ approved goods in the selected set.
\citet{AzizBrCo17} showed that the \emph{proportional approval voting (PAV)} rule always outputs a set of goods that satisfies EJR.
In cake sharing, \citet[Sec.~7]{BeiLuSu24} adapted EJR by imposing the condition for every positive \emph{real number} $t$, and proved that the resulting notion is satisfied by the \emph{maximum Nash welfare (MNW)} rule.\footnote{They also noted that JR does not admit a natural analog for cake sharing, since there is no discrete unit of cake.}
Can we unify the two versions of EJR for our generalized setting in such a way that the guaranteed existence is maintained?\footnote{As further evidence for the generality of our setting, we remark that, as \citet[Sec.~1.2]{BeiLuSu24} pointed out, cake sharing itself generalizes another collective choice setting called \emph{fair mixing} \citep{AzizBoMo20}.}

\begin{table*}[!ht]
\centering
\begin{tabular}{c  c  c  c}
\toprule
& \textbf{GreedyEJR-M} & \textbf{Gen.~MES} & \textbf{Gen.~PAV} \\
\midrule
\EJRM & \yes & \no & \no \\
\EJROne & \yes & \yes & \yes \\
Proportionality degree & $\floor{t} \cdot \left( 1 - \frac{\floor{t} + 1}{2t} \right) \approx \frac{t}{2}$ & $\left[ \frac{t - 2 + 1/t}{2}, \frac{\ceiling{t} + 1}{2} \right] \approx \frac{t}{2}$ & $> t-1$ \\
\midrule
Indivisible-goods EJR & \yes$^*$ & \yes$^*$ & \yes$^*$ \\
Cake EJR & \yes & \yes & \no \\
\midrule
Polynomial-time computation & ? & \yes & \no$^*$ \\
\bottomrule
\end{tabular}
\caption{Overview of our results.
The check mark (\yes) indicates that the rule satisfies the property; the cross mark (\no) indicates that it does not.
Entries marked by an asterisk follow from known results in multiwinner voting; the entry on the computation of Generalized PAV relies on the assumption that P $\ne$ NP. We also show that the proportionality degree implied by EJR-M and EJR-1 is $\floor{t} \cdot \left( 1 - \frac{\floor{t} + 1}{2t} \right)$ and $\frac{t-2+1/t}{2}$, which are both approximately $t/2$, respectively.}
\label{table:summary}
\end{table*}

\subsection{Our Contributions}

In \Cref{sec:EJR-notions}, we introduce two variants of EJR suitable for the mixed-goods setting.
The stronger variant, \emph{EJR for mixed goods (EJR-M)}, imposes the EJR condition for any positive real number $t$ whenever a $t$-cohesive group commonly approves a resource of size \emph{exactly} $t$.
The weaker variant, \emph{EJR up to~$1$ (EJR-1)}, again considers the condition for every positive real number $t$ but only requires that some member of a $t$-cohesive group receives utility greater than $t-1$.
While \mbox{EJR-M} reduces to the corresponding notion of EJR in both multiwinner voting and cake sharing, and therefore offers a unification of both versions, EJR-1 does so only for multiwinner voting.
We then extend three multiwinner voting rules to our setting: GreedyEJR, the method of equal shares (MES), and proportional approval voting (PAV).
We show that \emph{GreedyEJR-M}, our generalization of GreedyEJR, satisfies EJR-M (and therefore EJR-1), which also means that an EJR-M allocation always exists.
On the other hand, we prove that our generalizations of the other two methods provide EJR-1 but not EJR-M.
Furthermore, while \mbox{GreedyEJR-M} and Generalized MES guarantee the cake version of EJR in cake sharing, Generalized PAV does not.

In \Cref{sec:proportionality-degree}, we turn our attention to the concept of \emph{proportionality degree}, which measures the average utility of the agents in a cohesive group \citep{Skowron21}.
We derive tight bounds on the proportionality degree implied by both EJR-M and EJR-1, with the EJR-M bound being slightly higher.
We also investigate the proportionality degree of the three rules from \Cref{sec:EJR-notions}; in particular, we find that Generalized PAV has a significantly higher proportionality degree than both GreedyEJR-M and Generalized MES.

An overview of our results can be found in \Cref{table:summary}.

\section{Preliminaries}
\label{sec:prelim}

Let $N = \{1,2,\dots,n\}$ be the set of agents.
In the mixed-goods setting, the resource~$R$ consists of a cake $C = [0, c]$ for some real number $c\ge 0$ and a set of indivisible goods $G = \{g_1,\dots,g_m\}$ for some integer $m\ge 0$.
Assume without loss of generality that $\max(c,m) > 0$.
A \emph{piece of cake} is a union of finitely many disjoint (closed) subintervals of $C$.
Denote by $\ell(I)$ the length of an interval $I$, that is, $\ell([x, y]) \coloneqq y - x$.
For a piece of cake $C'$ consisting of a set of disjoint intervals $\mathcal{I}_{C'}$, we let $\ell(C') \coloneqq \sum_{I \in \mathcal{I}_{C'}} \ell(I)$.
A~\emph{bundle}~$R'$ consists of a (possibly empty) piece of cake $C'\subseteq C$ and a (possibly empty) set of indivisible goods $G'\subseteq G$; the \emph{size} of such a bundle $R'$ is $s(R') \coloneqq \ell(C') + |G'|$.
We sometimes write $R' = (C', G')$ instead of $R' = C'\cup G'$.

We assume that the agents have \emph{approval} preferences (also known as \emph{dichotomous} or \emph{binary}), i.e., each agent $i\in N$ approves a bundle $R_i = (C_i, G_i)$ of the resource.\footnote{Approval preferences can be given explicitly as part of the input for algorithms, so we do not need the cake-cutting query model of \citet{RobertsonWe98}.
In particular, the cake preferences can be described by the endpoints of the cake intervals approved by each agent.
}
The utility of agent~$i$ for a bundle $R'$ is given by $u_i(R') \coloneqq s(R_i\cap R') = \ell(C_i\cap C') + |G_i\cap G'|$.
Let $\alpha\in (0,c+m]$ be a given parameter, and assume that a bundle~$A$ with $s(A)\le \alpha$ can be chosen and collectively allocated to the agents;\footnote{Instead of the variable $k$ as in multiwinner voting, we use $\alpha$, as this variable may not be an integer in our setting. This is consistent with the notation used by \citet{BeiLuSu24} for cake sharing.} we also refer to an allocated bundle as an \emph{allocation}.
(Note that we allow $s(A)\le \alpha$ rather than requiring $s(A)= \alpha$; this is a slight deviation from the standard multiwinner voting model.)
An \emph{instance} consists of the resource~$R$, the agents~$N$ and their approved bundles $(R_i)_{i\in N}$, and the parameter $\alpha$.
We say that an instance is a \emph{cake instance} if it does not contain indivisible goods (i.e., $m = 0$), and an \emph{indivisible-goods instance} if it does not contain cake (i.e., $c = 0$).\footnote{When $c = 0$ the cake consists of a single point, which yields utility~$0$ to every agent, so we may ignore it.}
An example instance is shown in \Cref{fig:example-instance}.

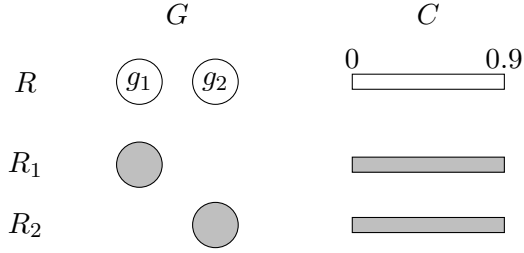
\begin{figure}
\centering
\begin{tikzpicture}[scale=1]
\node at (4.5,6.9) {$G$};
\node at (7.8,6.9) {$C$};

\draw (4,6) circle [radius=0.3];
\node at (4,6) {$g_1$};
\draw (5,6) circle [radius=0.3];
\node at (5,6) {$g_2$};
\draw (6.8,5.9) rectangle (8.8,6.1);
\node at (6.8,6.3) {$0$};
\node at (8.8,6.3) {$0.9$};
\node at (2.5,6) {$R$};

\draw[fill=gray!50] (4,4.9) circle [radius=0.3];
\draw[fill=gray!50] (6.8,4.8) rectangle (8.8,5);
\node at (2.5,4.9) {$R_1$};

\draw[fill=gray!50] (5,4.1) circle [radius=0.3];
\draw[fill=gray!50] (6.8,4) rectangle (8.8,4.2);
\node at (2.5,4.1) {$R_2$};
\end{tikzpicture}
\caption{A mixed-goods instance with two agents $N = \{1,2\}$, two indivisible goods $G = \{g_1,g_2\}$, a cake $C$ of length $0.9$, and $\alpha = 2$.
Agent~$1$ approves $R_1 = \{g_1\}\cup C$, while agent~$2$ approves $R_2 = \{g_2\}\cup C$.
If the allocation $A = \{g_1,g_2\}$ is chosen, both agents receive a utility of $1$.}
\label{fig:example-instance}
\end{figure}

A \emph{mechanism} or \emph{rule} $\mathcal{M}$ maps any instance to an allocation of the resource.
For any property $P$ of allocations, we say that a rule~$\mathcal{M}$ satisfies property $P$ if for every instance, the allocation output by $\mathcal{M}$ satisfies $P$.
An example of a rule is the \emph{maximum Nash welfare (MNW)} rule, which returns an allocation $A$ that maximizes the product $\prod_{i\in N}u_i(A)$ of the agents' utilities.\footnote{Ties can be broken arbitrarily except when the highest possible product is $0$.
In this exceptional case, the MNW rule first gives positive utility to a set of agents of maximal size and then maximizes the product of utilities for the agents in this set.}

\section{EJR Notions and Rules}
\label{sec:EJR-notions}

In order to reason about \emph{extended justified representation (EJR)}, an important concept is that of a cohesive group.
For any positive real number $t$, a set of agents $N^*\subseteq N$ is said to be \emph{$t$-cohesive} if $|N^*| \ge t\cdot n/\alpha$ and $s(\bigcap_{i \in N^*} R_i) \ge t$.
For an indivisible-goods instance, \citet{AzizBrCo17} defined EJR as follows: an allocation~$A$ satisfies EJR if for every positive integer $t$ and every $t$-cohesive group of agents $N^*$, at least one agent in~$N^*$ receives utility at least $t$.
\citet{BeiLuSu24} adapted this axiom to cake sharing by considering every positive real number $t$ instead of only positive integers.\footnote{Note that the indivisible-goods version with positive integers~$t$ may be meaningless in the cake setting, e.g., if the entire cake has length less than $1$. More generally, the restriction to positive integers $t$ is unnatural for cake, as there is no discrete unit of cake.}
To distinguish between these two versions of EJR, as well as from versions for mixed goods that we will define next, we refer to the two versions as \emph{indivisible-goods EJR} and \emph{cake EJR}, respectively.

A first attempt to define EJR for mixed goods is to simply use the cake version.
However, as we will see shortly, the resulting notion is too strong.
Hence, we relax it by lowering the utility threshold.

\begin{definition}[EJR-$\beta$]
\label{def:EJR-beta}
Let $\beta \ge 0$.
Given an instance, an allocation $A$ with $s(A) \le \alpha$ is said to satisfy \emph{extended justified representation up to $\beta$ (EJR-$\beta$)} if for every positive real number $t$ and every $t$-cohesive group of agents $N^*$, it holds that $u_j(A) > t-\beta$ for some $j\in N^*$.\footnote{For $\beta = 1$, \citet{PetersPiSk21} considered a somewhat similar notion called ``EJR up to one project'' in the setting of participatory budgeting with indivisible projects.}
\end{definition}

\begin{proposition}
\label{prop:EJR-beta}
For each constant $\beta \in [0,1)$, there exists an indivisible-goods instance in which no allocation satisfies EJR-$\beta$.
This remains true even if we relax the inequality $u_j(A) > t-\beta$ in \Cref{def:EJR-beta} to $u_j(A) \ge t-\beta$.
\end{proposition}

\begin{proof}
We work with the weaker condition $u_j(A) \ge t-\beta$.
Fix $\beta\in[0,1)$, and choose a rational constant $\beta' \in (\beta, 1)$.
Consider an indivisible-goods instance with integers $n$ and~$\alpha$ such that $\alpha = \beta'\cdot n$, and assume that all agents approve disjoint nonempty subsets $G_i$ of goods.
Each individual agent forms a $\beta'$-cohesive group, so in an \mbox{EJR-$\beta$} allocation, every agent must receive utility at least $\beta' - \beta > 0$.
Hence, any EJR-$\beta$ allocation necessarily includes at least one good from each approval set $G_i$, and must therefore contain at least $n$ goods in total.
However, since $\alpha = \beta'\cdot n < n$, no allocation can satisfy EJR-$\beta$.
\end{proof}

\Cref{prop:EJR-beta} raises the question of whether EJR-1 can always be satisfied.
We will answer this question in the affirmative in \Cref{sec:GreedyEJR-M}.
Before that, we introduce EJR-M, another variant of EJR tailored to mixed goods.
The intuition behind EJR-M is that a $t$-cohesive group of agents should be able to claim a utility of $t$ for some member only when there exists a commonly approved resource of size \emph{exactly}~$t$.
This rules out such cases as in the proof of \Cref{prop:EJR-beta}, where a group can effectively claim utility higher than $t$ due to the indivisibility of the goods.

\begin{definition}[EJR-M]
\label{def:EJR-M}
Given an instance, an allocation~$A$ with $s(A) \le \alpha$ is said to satisfy \emph{extended justified representation for mixed goods (EJR-M)} if the following holds:

For every positive real number $t$ and every $t$-cohesive group of agents $N^*$ for which there exists $R^*\subseteq R$ such that $s(R^*) = t$ and $R^*\subseteq R_i$ for all $i \in N^*$, it holds that $u_j(A) \ge t$ for some $j\in N^*$.
\end{definition}

Note that for indivisible-goods instances, the condition $s(R^*) = t$ can only hold for integers $t$, so EJR-M reduces to indivisible-goods EJR.
Likewise, for cake instances, if a group is $t$-cohesive then a commonly approved subset of size exactly $t$ always exists, so EJR-M reduces to cake EJR.
Hence, EJR-M unifies EJR from both settings.

\begin{proposition}
\label{prop:EJR-M-floor}
Let $t$ be a positive real number.
For an EJR-M allocation $A$ and a $t$-cohesive group of agents $N^*$, it holds that $u_j(A) \ge \floor{t}$ for some $j\in N^*$.
\end{proposition}

\begin{proof}
Let $R^* = \bigcap_{i \in N^*} R_i$, so $s(R^*) \ge t$, and let $m^*$ be the number of indivisible goods in $R^*$.
If $m^* \ge \floor{t}$, then by \Cref{def:EJR-M}, there exists $j\in N^*$ such that $u_j(A) \ge \floor{t}$.
Else, $m^* < \floor{t}$, which means that $R^*$ contains a piece of cake of length at least $t - m^*$.
In this case, by considering the $m^*$ indivisible goods and a piece of cake of length exactly $t-m^*$ commonly approved by all agents in $N^*$, \Cref{def:EJR-M} implies the existence of $j\in N^*$ such that $u_j(A) \ge t \ge \floor{t}$.
\end{proof}

Since $\floor{t} > t-1$ for every real number $t$, we have the following corollary.

\begin{corollary}
\label{cor:EJR-M-EJR-1}
EJR-M implies EJR-1.
\end{corollary}

For indivisible-goods instances, EJR-1 reduces to indivisible-goods EJR, since for every positive real number $t$, the smallest integer greater than $t-1$ is $\floor{t}$.
On the other hand, for cake instances, EJR-1 is weaker than cake EJR.

In the cake setting, \citet{BeiLuSu24} proved that the MNW rule satisfies cake EJR.
However, in the indivisible-goods setting, the fact that MNW tries to avoid giving utility~$0$ to any agent at all costs means that it sometimes attempts to help individual agents at the expense of large deserving groups.
This is formalized in the following proposition.

\begin{proposition}
\label{prop:MNW-EJR}
For any constant $\beta \ge 0$, there exists an indivisible-goods instance in which no MNW allocation satisfies EJR-$\beta$.
\end{proposition}

\begin{proof}
It suffices to prove the statement for every positive integer $\beta$.
Indeed, once we have this, then for any nonnegative real number $\beta'$, there exists a positive integer $\beta > \beta'$.
Since EJR-$\beta'$ implies EJR-$\beta$, in an instance in which no MNW allocation satisfies EJR-$\beta$, there also does not exist an MNW allocation satisfying EJR-$\beta'$.

Fix a positive integer $\beta$, and let $\gamma = \beta + 2$.
Consider an indivisible-goods instance with $n = \gamma^2 + \gamma$ agents, $m = 2\gamma$ goods, and $\alpha = \gamma+1$.
The first $\gamma^2$ agents all approve goods $g_1,\dots,g_{\gamma}$, while agent $\gamma^2 + i$ only approves good $g_{\gamma + i}$ for $1\le i\le \gamma$.
Notice that the first $\gamma^2$ agents form a $\gamma$-cohesive group, so at least one of them must receive utility no less than $\gamma - \beta = 2$ in an EJR-$\beta$ allocation.
In particular, at least two goods among $g_1,\dots,g_{\gamma}$ must be chosen.
However, every MNW allocation contains $g_{\gamma+1}, g_{\gamma+2}, \dots, g_{2\gamma}$ along with exactly one of $g_1,\dots,g_{\gamma}$.
It follows that no MNW allocation satisfies EJR-$\beta$.
\end{proof}

\subsection{GreedyEJR-M}
\label{sec:GreedyEJR-M}

\Cref{prop:MNW-EJR} implies that the MNW rule cannot guarantee EJR-M or EJR-1 in the indivisible-goods setting, let alone in the mixed-goods setting.
We show next that a greedy approach can be used to achieve these guarantees.
The rule that we use is an adaptation of the \emph{GreedyEJR} rule from the indivisible-goods setting \citep{BredereckFaKa19,PetersPiSk21,ElkindFaIg22}; we therefore call it \emph{GreedyEJR-M} and describe it below.

\begin{framed}
\noindent
\textbf{GreedyEJR-M} \\

\noindent
\emph{Step~1:} Initialize $N' = N$ and $R' = \emptyset$.\\

\noindent
\emph{Step~2:} Let $t^*$ be the largest \emph{nonnegative} real number for which there exist $\emptyset \neq N^*\subseteq N'$ and $R^*\subseteq R$ such that $N^*$ is a $t^*$-cohesive group, $R^*\subseteq R_i$ for all $i \in N^*$, and $s(R^*) = t^*$.
Consider any such pair $(N^*,R^*)$.
Remove $N^*$ from $N'$ and add the part of $R^*$ that is not already in $R'$ to $R'$. \\

\noindent
\emph{Step~3:} If $N' = \emptyset$, return $R'$. Else, go back to Step~2.
\end{framed}

\begin{example}
\label{ex:running}
Consider the instance in \Cref{fig:example-instance}.
We have $n/\alpha = 1$, and Step~2 of GreedyEJR-M chooses $t^* = 1$, along with (as one possibility) $N^* = \{1\}$ and $R^* = \{g_1\}$.
We are left with $N' = \{2\}$, and the next iteration of Step~2 chooses $t^* = 1$, $N^* = \{2\}$, and $R^* = \{g_2\}$.
Finally, the rule returns $R' = \{g_1,g_2\}$.
\end{example}

\begin{theorem}
\label{thm:greedyEJR-M}
The GreedyEJR-M rule satisfies EJR-M (and therefore EJR-1).
\end{theorem}

\begin{proof}
By \Cref{cor:EJR-M-EJR-1}, it suffices to prove the claim for EJR-M.
We break the proof into the following four parts.
\begin{itemize}
\item \emph{The procedure is well-defined.}
To this end, we must show that the largest nonnegative real number $t^*$ in Step~2 always exists.
Observe that for each nonempty group of agents $X\subseteq N'$, the set
\begin{equation*}
T_X \coloneqq \biggl\{\, t \ge 0 \,\biggm|\, |X| \ge t \cdot \frac{n}{\alpha} \text{ and  there exists } Y \subseteq \bigcap_{i \in X} R_i \text{ with } s(Y) = t  \,\biggr\}
\end{equation*}
is a union of a finite number of (possibly degenerate) closed intervals, and is nonempty because $0\in T_X$. Therefore, $T_X$ has a maximum.
The value $t^*$ chosen in Step~2 is then the largest among the maxima of $T_X$ across all nonempty $X\subseteq N'$.
\item \emph{The procedure always terminates.}
This is because each iteration of Step~2 removes at least one agent from $N'$.
\item \emph{The procedure returns an allocation $R'$ with $s(R') \le \alpha$.}
Indeed, if an iteration of Step~2 uses value $t^*$, it removes\footnote{If $t^* = 0$, the iteration still removes at least one agent from $N'$, but we do not need this fact here.} at least $t^*\cdot n/\alpha$ agents from $N'$ and adds a resource of size at most~$t^*$ to $R'$.
Since only $n$ agents can be removed in total, the added resource has size at most $\alpha$.
\item \emph{The returned allocation $R'$ satisfies EJR-M.}
Assume for contradiction that for some group $X$, \Cref{def:EJR-M} fails for $X$ and parameter $t$.
Consider the moment after the procedure removed the last group with parameter $t^* \ge t$.
If no agent in $X$ has been removed, the procedure should have removed $X$ with parameter $t$, a contradiction.
Else, some agent $j\in X$ has been removed.
In this case, the procedure guarantees that $u_j(R') \ge t$, which means that $X$ satisfies \Cref{def:EJR-M} with parameter $t$, again a contradiction. \qedhere
\end{itemize}
\end{proof}
\subsection{\generalizedMES}
\label{sec:Gen-MES}

Despite the strong representation guarantee provided by GreedyEJR-M, the rule does not admit an obvious polynomial-time implementation.\footnote{Indeed, determining $t^*$ in Step~2 of GreedyEJR-M potentially requires inspecting an exponential number of subsets $N^*\subseteq N'$.}
In the indivisible-goods setting, \citet{PetersSk20} introduced the \emph{\RuleX (\rulex)}, originally known as \emph{Rule~X}, and showed that it satisfies indivisible-goods \EJR{} and runs in polynomial time.
We now extend their rule to our mixed-goods setting.
At a high level, in \emph{\generalizedMESAbbr}, each agent is given a budget of~$\alpha / n$, which can be spent on buying the resource---each piece of cake has cost equal to its length whereas each indivisible good costs~$1$.
In each step, a piece of cake or an indivisible good that incurs the smallest cost per utility for agents who approve it is chosen, and these agents pay as equally as possible to cover the cost of the chosen resource.
The rule stops once no more cake or indivisible good is affordable.
Note that when the resource consists only of indivisible goods, \generalizedMESAbbr is equivalent to the original \rulex of \citet{PetersSk20}.

\begin{framed}
\noindent
\textbf{\generalizedMESAbbr} \\

\noindent
\emph{Step~1:} Initialize $R' = (C', G') = (\emptyset, \emptyset)$ and $b_i = \alpha/n$ for each $i \in N$.\\

\noindent
\emph{Step~2:} Divide the remaining cake~$C$ into intervals $I_1, \dots, I_k$ so that each agent approves each interval either entirely or not at all.
For each interval $ I_j = [x_0, x_1]$, $x \in (x_0, x_1]$, and $\rho \ge 0$, we say that $I_j$ is \emph{$(x, \rho)$-affordable} if
\[
\sum_{i \in N_{I_j}} \min(b_i, (x-x_0)\cdot \rho) = x-x_0,
\]
where $N_{I_j}\subseteq N$ denotes the set of remaining agents who approve $I_j$.
Similarly, for each remaining good $g \in G$ and $\rho \ge 0$, we say that $g$ is \emph{$\rho$-affordable} if
\[
\sum_{i \in N_g} \min(b_i,\rho) = 1,
\]
where $N_g\subseteq N$ denotes the set of remaining agents who approve $g$.
\\

\noindent
\emph{Step~3:} If for every $\rho$, no $\rho$-affordable good or $(x,\rho)$-affordable piece of cake exists, return $R'$.

Else, take either an interval $I_j$ with the smallest $\rho$ along with the largest $x$ such that $I_j$ is $(x, \rho)$-affordable, or a good $g$ with the smallest $\rho$ such that $g$ is $\rho$-affordable, depending on which $\rho$ is smaller.
In the former case, deduct $\min(b_i, (x-x_0)\cdot\rho)$ from $b_i$ for each $i\in N_{I_j}$, and set $C = C \setminus [x_0, x]$ and $C' = C' \cup [x_0, x]$.
In the latter case, deduct $\min(b_i, \rho)$ from $b_i$ for each $i\in N_g$, and set $G = G \setminus \{g\}$ and $G' = G'  \cup \{g\}$.
Remove all agents who have run out of budget from $N$, and go back to Step~2.
\end{framed}

\begin{example}\label{ex:Gen-MES}
For the instance in \Cref{fig:example-instance}, each agent starts with a budget of $\alpha/n = 1$.
The first iteration of Step~2 selects the entire cake (with $\rho = 1/2$), and each agent pays $0.9/2 = 0.45$ for this cake.
Since neither agent has enough budget left to buy the indivisible good that she approves (which costs $1$), the procedure terminates with only the cake.
\end{example}

In the instance above, each agent on her own is $1$-cohesive and approves a subset of the resource of size exactly $1$, so the only EJR-M allocation is $\{g_1, g_2\}$.
In particular, the allocation chosen by Generalized MES is not EJR-M.

\begin{proposition}
\label{prop:GMES-EJR-M}
\generalizedMESAbbr does not satisfy EJR-M.
\end{proposition}

Nevertheless, we prove that \generalizedMESAbbr satisfies EJR-1 and, moreover, can be implemented efficiently.
We do not prove the result directly, but instead introduce a stronger notion than \EJROne and show that \generalizedMESAbbr satisfies even this stronger notion.

\begin{definition}[\expandafter\MakeUppercase\strongEJROne]\label{def:strong-EJR-1}
Given an instance, an allocation~$A$ with $s(A) \leq \alpha$ satisfies \emph{\strongEJROne} if for every positive real number~$t$ and every $t$-cohesive group of agents~$N^*$, at least one of the following two conditions holds:
\begin{itemize}
\item $\bigcap_{i \in N^*} C_i \subseteq A$ \emph{and} there exists some~$j \in N^*$ with $u_j(A) > t-1$; or
\item there exists some~$j \in N^*$ with $u_j(A) \geq t$.
\end{itemize}
\end{definition}

It follows directly from the definition that \strongEJROne implies \EJROne, since in either case, there exists some agent in the $t$-cohesive group who receives utility greater than~$t-1$.
\expandafter\MakeUppercase\strongEJROne is more demanding than \EJROne in that if no agent in a $t$-cohesive group receives utility at least~$t$, then the cake commonly approved by all agents in that group must be entirely included in the allocation.

On the one hand, like \EJROne, for indivisible-goods instances, \strongEJROne reduces to indivisible-goods \EJR{}, since (i) for any cohesive group~$N^*$, it holds trivially that $\bigcap_{i \in N^*} C_i = \emptyset \subseteq A$, and (ii) for every positive real number~$t$, the smallest integer greater than~$t-1$ is~$\floor{t}$.
On the other hand, unlike \EJROne, for cake instances, \strongEJROne reduces to cake \EJR{}, since for any $t$-cohesive group~$N^*$ such that $\bigcap_{i \in N^*} C_i \subseteq A$, it must be the case that some agent (in fact, all agents) in the cohesive group receives a utility of at least~$t$.
Due to its somewhat unintuitive definition, we view \strongEJROne as a technical strengthening of \EJROne rather than a practical one.

It is tempting to believe that \strongEJROne lies between \EJRM and \EJROne.
This is, however, not the case.
The proof of this claim can be found in \Cref{app:no-logic-relation}.

\begin{proposition}\label{prop:EJR-M-strong-EJR-1-no-logic-relation}
Neither \EJRM nor \strongEJROne implies each other.
Moreover, GreedyEJR-M does not satisfy \strongEJROne.
\end{proposition}

As we will see later, Generalized PAV does not satisfy cake \EJR{} for cake instances (\Cref{prop:GPAV-cake-EJR}), so it does not satisfy \strongEJROne either.
We now show that \generalizedMESAbbr satisfies \strongEJROne.

\begin{theorem}
\label{thm:GMES}
\generalizedMESAbbr satisfies \strongEJROne and can be implemented in polynomial time.
\end{theorem}

\begin{proof}
First, observe that for each interval $I_j$, provided that $N_{I_j} \ne \emptyset$, the value of $\rho$ chosen in Step~3 will be $\rho = 1/|N_{I_j}|$, and the value of $x$ will be either $x_1$ or the smallest value such that $(x-x_0)\cdot \rho = b_i$ for some $i\in N_{I_j}$, whichever is smaller.
For each indivisible good $g$, the value of $\rho$ can also be computed in polynomial time.\footnote{See Footnote~9 in the extended version of \citet{PetersSk20}'s work.}
After each iteration of Step~3, if the procedure has not terminated, at least one of the following occurs: an entire interval~$I_j$ is removed from~$C$, an indivisible good $g$ is removed from~$G$, or one or more agents run out of budget.
Hence, the procedure can be implemented in polynomial time.

Next, note that whenever a resource of some size $y$ is added to $R'$, the agents together pay a total of $y$.
Since the agents have a total starting budget of $n\cdot(\alpha/n) = \alpha$, \generalizedMESAbbr returns an allocation of size at most $\size$.

We now show that the returned allocation~$R' = (C', G')$ satisfies \strongEJROne.
To begin with, assume for contradiction that for some real number~$t > 0$, there exists a $t$-cohesive group~$N'$ with $\bigcap_{i \in N'} C_i \subseteq C'$ and $u_i(R') \leq t - 1$ for all~$i \in N'$.
In particular, all cake commonly approved by the agents in~$N'$ has already been included in~$R'$ by \generalizedMESAbbr.
Due to $t$-cohesiveness, there must still exist a good from $\bigcap_{i \in N'} G_i$ left in~$G$.
Moreover, $\sum_{i \in N'} b_i < 1$ at the termination of the procedure---indeed, if $\sum_{i \in N'} b_i \ge 1$, then the agents in $N'$ together have enough budget to afford a commonly approved good, so the procedure should not have terminated without choosing this good.
Since $|N'| \ge t \cdot \frac{n}{\alpha}$, we have $b_i < \frac{\size}{tn}$ for some $i \in N'$.
In particular, agent~$i$ receives a utility of at most $t-1$ but has spent more than $\frac{\size}{n} - \frac{\size}{tn} = \frac{\size}{n} \cdot \frac{t-1}{t}$.
Thus, the cost per utility for $i$ is strictly greater than $\frac{1}{t-1}\left(\frac{\size}{n} \cdot \frac{t-1}{t}\right) = \frac{\size}{tn}$.
Now, consider the first cake interval or indivisible good added in Step~3 for which the cost per utility for some agent in~$N'$ exceeds~$\frac{\size}{tn}$; the existence of such an interval or good follows from the previous sentence.
Note that the value of~$\rho$ in this step must be larger than $\frac{\size}{tn}$.
However, since this is the first step in which an agent from $N'$ pays more than $\frac{\size}{tn}$ per utility and the utility of each agent in $N'$ is at most $t-1$, each agent in $N'$ must have budget at least $\frac{\size}{n} - \frac{\size}{tn}\cdot (t-1) = \frac{\size}{tn}$ remaining before this step.
Since $|N'| \ge \frac{tn}{\alpha}$, before this step, there is still an indivisible good from $\bigcap_{i \in N'} G_i$ which is $\rho$-affordable for some $\rho \le \frac{\size}{tn}$.
This contradicts the fact that \generalizedMESAbbr chooses a resource with $\rho > \frac{\size}{tn}$.
Hence, if $\bigcap_{i \in N'} C_i \subseteq C'$, then $u_i(R') > t-1$ for some $i\in N'$.

Next, assume for contradiction that for some real number~$t > 0$, there exists a $t$-cohesive group~$N'$ with $(\bigcap_{i \in N'} C_i) \setminus C' \neq \emptyset$ and $u_i(R') < t$ for all~$i \in N'$.
Let $\delta > 0$ be such that $u_i(R') < t - \delta$ for all~$i \in N'$.
Since $(\bigcap_{i \in N'} C_i) \setminus C' \neq \emptyset$ and this piece of cake is not affordable at the end, the budget of all agents in~$N'$ must have run out.
Hence, the cost per utility for every agent in~$N'$ is strictly greater than $\frac{\size}{(t - \delta) n}$.
Now, consider the first cake interval or indivisible good added in Step~3 for which the cost per utility for some agent in~$N'$ exceeds~$\frac{\size}{tn}$; the existence of such an interval follows from the previous sentence since $\frac{\size}{(t - \delta) n} > \frac{\size}{tn}$.
Note that the value of $\rho$ in this step must be larger than~$\frac{\size}{tn}$.
However, since this is the first step in which an agent from~$N'$ pays strictly more than~$\frac{\size}{tn}$ per utility and the utility of each agent in~$N'$ is at most~$t - \delta$, each agent in~$N'$ must have budget at least $\frac{\size}{n} - \frac{\size(t - \delta)}{tn} = \frac{\size\delta}{tn} > 0$ remaining before this step.
Since $|N'| \geq t \cdot \frac{n}{\alpha}$, before this step there is still an interval $I_j = [x_0, x_1]$ from $\bigcap_{i \in N'} C_i$ which is $(x_0 + \varepsilon, \rho)$-affordable for some $\varepsilon > 0$ and $\rho \leq \frac{\size}{tn}$.
This contradicts the fact that \generalizedMESAbbr chooses a resource with $\rho > \frac{\size}{tn}$.
\end{proof}

As a corollary, we have the following result.

\begin{corollary}
\label{prop:GMES-cake-EJR}
For cake instances, \generalizedMESAbbr satisfies cake \EJR{}.
\end{corollary}

\subsection{Generalized PAV}

In the indivisible-goods setting, a well-studied rule is \emph{proportional approval voting (PAV)}, which chooses an allocation $R'$ that maximizes $\sum_{i\in N}H_{u_i(R')}$, where $H_x \coloneqq 1 + \frac{1}{2} + \dots + \frac{1}{x}$ is the $x$-th harmonic number.
We now show how to generalize PAV to the mixed-goods setting.
To this end, we will use a continuous extension of the harmonic numbers based on the digamma function, defined as $H_x \coloneqq \psi(x+1) + \gamma$ for all real numbers $x\ge 0$, where $\psi$ is the digamma function and $\gamma$ is the Euler--Mascheroni constant.
A known fact about the digamma function (e.g., \citep[p.~259, Formula~6.3.16]{AbramowitzSt72}) implies that
\[
H_x = \sum_{k = 1}^\infty \frac{x}{k(x+k)} = \sum_{k = 1}^\infty \left(\frac{1}{k} - \frac{1}{k+x}\right)
\]
for each real number $x\ge 0$; in particular, these infinite sums converge.
It is clear from the definition that the generalized harmonic numbers indeed extend the original harmonic numbers,\footnote{The extension that we use is in a certain sense the canonical extension of the harmonic numbers. 
Indeed, $\psi$ is known to be the unique monotonic solution of the functional equation $F(x+1) = F(x) + 1/x$ that satisfies $F(1) = -\gamma$.
If we tried to extend the harmonic numbers using, say, the piecewise linear extension $\widehat{H}$, we would have $\widehat{H}_{0.5} = 0.5$ and $\widehat{H}_{1.5} = 1.25$, so $\widehat{H}_{1.5} - \widehat{H}_{0.5} = 0.75 \ne 1/1.5$.
} and that $H_x > H_y$ for all $x > y \ge 0$.
Moreover, $H_{x+1} - H_x = \frac{1}{x+1}$ for all $x \ge 0$.

\begin{definition}[Generalized PAV]
The \emph{Generalized PAV} rule selects an allocation~$R'$ with $s(R') \le \size$ that maximizes  $\sum_{i\in N} H_{u_i(R')}$.
\end{definition}

For ease of notation, we let $H(R') \coloneqq \sum_{i\in N} H_{u_i(R')}$ for any allocation $R'$, and call $H(R')$ the \emph{GPAV-score} of $R'$.
Given the instance in \Cref{fig:example-instance}, since $H_{1.9} + H_{0.9} > 1.45 + 0.93 > 1 + 1 = H_1 + H_1$, Generalized PAV selects the entire cake together with one of the indivisible goods.
As the only EJR-M allocation in this instance is $\{g_1, g_2\}$, the allocation selected by Generalized PAV is not EJR-M.

\begin{proposition}
\label{prop:GPAV-EJR-M}
Generalized PAV does not satisfy EJR-M.
\end{proposition}

To show that Generalized PAV satisfies EJR-1, we establish a useful lemma on the growth rate of the generalized harmonic numbers.

    \begin{lemma}
    \label{lem:harmonic-growth}
    For any $x \in (0, \infty)$ and $y \in [0,1]$, it holds that
    $    H_{x + y} - H_{x} \le \frac{y}{x+y}$.
    \end{lemma}

    \begin{proof}
    First, note that for any positive integer $r$,
    \begin{align*}
        \sum_{k = 1}^r \frac{x+y}{k(x+y+k)} - \sum_{k = 1}^r  \frac{x}{k(x+k)}
        &=  \sum_{k = 1}^r \left(\frac{1}{k} - \frac{1}{ k + x + y}\right) - \sum_{k = 1}^r \left(\frac{1}{k} - \frac{1}{k + x} \right) \\
        &= \sum_{k = 1}^r \left(\frac{1}{k + x} - \frac{1}{k + x + y}\right) \\
        &=  \sum_{k = 1}^r \frac{y}{(k + x)(k + x + y)} \\
        &\le \sum_{k = 1}^r \frac{y}{(k + x + y - 1)(k + x + y)}.
    \end{align*}
    Hence, we have
    \begin{align*}
        H_{x + y} - H_{x}
        &=  \lim_{r \to \infty} \sum_{k = 1}^r \frac{x+y}{k(x+y+k)} - \lim_{r \to \infty} \sum_{k = 1}^r \frac{x}{k(x+k)} \\
        &= \lim_{r \to \infty} \left( \sum_{k = 1}^r \frac{x+y}{k(x+y+k)} - \sum_{k = 1}^r  \frac{x}{k(x+k)} \right) \\
        &\le \lim_{r \to \infty} \sum_{k = 1}^r \frac{y}{(k + x + y - 1)(k + x + y)} \\
        &= \sum_{k = 1}^\infty \frac{y}{(k + x + y - 1)(k + x + y)} \\
        &= \sum_{k = 1}^\infty \left( \frac{y}{k+x+y-1} - \frac{y}{k+x+y} \right)
        = \frac{y}{x+y},
    \end{align*}
    where the inequality follows from the previous paragraph and the last sum telescopes.
    \end{proof}

\begin{theorem}
\label{thm:GPAV}
Generalized PAV satisfies EJR-1.
\end{theorem}

\begin{proof}
Let $R' = (C', G')$ be a Generalized PAV allocation.
By adding a piece of cake approved by no agent to the resource~$R$ as well as~$R'$ if necessary, we may assume without loss of generality that $s(R') = \alpha$.
Assume also that the cake~$C'$ is represented by the interval $[0, c']$.
Whenever $x + y > c'$, the interval $[x, x + y]$ refers to $[x, c'] \cup [0, x + y - c']$, i.e., we cyclically wrap around the cake $C'$.\footnote{For example, if $c' = 2$, $x = 1.2$, and $y = 1$, then $[x, x+y]$ refers to $[1.2,2]\cup [0,0.2]$ since $x+y-c' = 0.2$.}

Suppose for contradiction that for some~$t > 0$, there exists a $t$-cohesive group~$N'$ with $u_i(R') \leq t-1$ for all~$i \in N'$.
Hence, there exists either a piece of cake of size~$1$ that is approved by all agents in~$N'$ but not contained in~$R'$, or an indivisible good with the same property.
We assume the latter case; the proof proceeds similarly in the former case.
Denote this good by~$g^*$, and let $G'' \coloneqq G' \cup \{g^*\}$ and $R'' \coloneqq (C', G'')$.
We have
\begin{align*}
H(R'') - H(R')
&\ge \sum_{i \in N'} \left(H_{u_i(R') + 1} - H_{u_i(R')}\right) \\
&= \sum_{i \in N'} \frac{1}{u_i(R') + 1} \\
&\ge \frac{\lvert N'\rvert^2}{\sum_{i \in N'}(u_i(R') + 1)}
\ge \frac{\lvert N'\rvert^2}{\lvert N'\rvert\cdot (t-1) + \lvert N'\rvert}
= \frac{|N'|}{t}
\ge \frac{n}{\size},
\end{align*}
where the second inequality follows from the inequality of arithmetic and harmonic means and the last inequality from the definition of a $t$-cohesive group.
In other words, adding~$g^*$ increases the GPAV-score of~$R'$ by at least~$n/\size$.

For each good~$g \in G$, denote by $N_g \subseteq N$ the set of agents who approve it.
For each~$g \in G''$, we have
\begin{align*}
H(R'') - H(R''\setminus \{g\})
&= \sum_{i \in N_g} \left(H_{u_i(R'')} - H_{u_i(R'') - 1}\right)
= \sum_{i \in N_g} \frac{1}{u_i(R'')},
\end{align*}
where $u_i(R'')-1 \ge 0$ always holds because every agent $i\in N_g$ approves $g\in G''$ and therefore has utility at least~$1$ for $R'' = (C', G'')$.
Letting $N_+$ consist of the agents~$i\in N$ with $u_i(R'') > 0$, we get
\begin{align}
\sum_{g \in G''} (H(R'') - H(R''\setminus \{g\}))
&= \sum_{g \in G''} \sum_{i \in N_g} \frac{1}{u_i(R'')} \nonumber \\
&= \sum_{i \in N_+} \sum_{g \in G'' \cap G_i} \frac{1}{u_i(R'')} \nonumber \\
&= \sum_{i \in N_+} \frac{u_i(G'')}{u_i(R'')} \label{eq:GPAV-difference-good} \\
&\le \sum_{i \in N_+} 1 \le n. \nonumber
\end{align}
If there is a good~$g \in G''$ such that $H(R'') - H(R'' \setminus \{g\}) < n/\alpha$ (clearly, $g \ne g^*$), we can replace~$g$ with~$g^*$ in~$R'$ and obtain a higher GPAV-score, contradicting the definition of~$R'$.
Hence, we may assume that $H(R'') - H(R'' \setminus \{g\}) \ge n/\alpha$ for every good~$g \in G''$.
It follows that
\[
n \ge \sum_{g \in G''} (H(R'') - H(R''\setminus \{g\})) \ge |G''|\cdot \frac{n}{\alpha}.
\]
Therefore, we have that $|G''| \le \alpha$, and so $c' \ge 1$.

Now, for any $x \in C'$, it holds that
\begin{align}
H(R'') - H(R'' \setminus [x, x + 1])
&= \sum_{i \in N} \left( H_{u_i(R'')} - H_{u_i(R'') - u_i([x, x + 1])} \right) \nonumber \\
&\le \sum_{i \in N_+} \frac{u_i([x, x + 1])}{u_i(R'')}, \label{eq:GPAV-difference-cake}
\end{align}
where the inequality follows from \Cref{lem:harmonic-growth}.
Using \eqref{eq:GPAV-difference-good} and \eqref{eq:GPAV-difference-cake}, we get
\begin{align}
\sum_{g \in G''} (H(R'') - H(R'' \setminus \{g\}))
&+ \int_{C'} (H(R'') - H(R''\setminus [x, x + 1])) \diff x \nonumber \\
&\le \sum_{i\in N_+}\frac{u_i(G'')}{u_i(R'')} + \int_{C'} \left(\sum_{i \in N_+}\frac{u_i([x, x + 1])}{u_i(R'')}\right) \diff x \nonumber \\
&= \sum_{i\in N_+}\frac{u_i(G'')}{u_i(R'')} + \sum_{i\in N_+}\left(\int_{C'} \frac{u_i([x, x + 1])}{u_i(R'')} \diff x\right) \nonumber \\
&= \sum_{i \in N_+} \left[\frac{1}{u_i(R'')}\left(u_i(G'') + \int_{C'} u_i([x, x + 1])\diff x\right) \right] \nonumber \\
&= \sum_{i \in N_+} \left[\frac{1}{u_i(R'')}\left(u_i(G'') +  u_i(C')\right)\right] \le \sum_{i\in N} 1 = n. \label{eq:GPAV}
\end{align}
Here, we have $\int_{C'} u_i([x, x + 1])\diff x = u_i(C')$ because
\begin{align*}
\int_{C'} u_i([x,x+1])\diff x
&= \int_{C'} \ell(C_i\cap [x,x+1])\diff x \\
&= \int_{C_i\cap C'}\ell([y-1,y])\diff y
= \int_{C_i\cap C'}1 \diff y
= \ell(C_i\cap C')
= u_i(C'),
\end{align*}
where the second equality holds because a point $y \in C_i$ belongs to the interval $[x, x + 1]$ if and only if $x \in [y - 1, y]$.

If it were the case that $H(R'') - H(R'' \setminus [x, x + 1]) \ge n/\alpha$ for every $x \in C'$, we would have
\begin{align*}
\sum_{g \in G''} (H(R'') - H(R''\setminus \{g\}))
&+ \int_{C'} (H(R'') - H(R''\setminus [x, x + 1])) \diff x\\
&\ge |G''|\cdot\frac{n}{\alpha} + c'\cdot \frac{n}{\alpha} = (\alpha+1)\cdot \frac{n}{\alpha} > n,
\end{align*}
a contradiction with \eqref{eq:GPAV}.
Thus, it must be that $H(R'') - H(R'' \setminus [x, x + 1]) < n/\alpha$ for some $x \in C'$.
By replacing the cake $[x, x + 1]$ in~$R'$ with the good~$g^*$, we therefore obtain a higher GPAV-score than that of~$R'$.
This yields the final contradiction and completes the proof.
\end{proof}

In contrast to \generalizedMESAbbr, Generalized PAV does not satisfy EJR in cake sharing.

\begin{proposition}
\label{prop:GPAV-cake-EJR}
For cake instances, Generalized PAV does not satisfy cake EJR.
\end{proposition}

To prove this statement, we use the following proposition.

\begin{proposition}[\citet{BeiLuSu24}]
\label{prop:maximize-EJR-M}
Let $f:\mathbb{R}_{\ge 0}\rightarrow [-\infty,\infty)$ be a strictly increasing function which is differentiable in $(0,\infty)$.
For cake sharing, if a rule that always chooses an allocation~$R'$ maximizing $\sum_{i\in N} f(u_i(R'))$ satisfies cake EJR, then there exists a constant $c$ such that $f'(x) = c/x$ for all $x\in (0,\infty)$.\footnote{This is Theorem~7.8 in their work.
Bei et al.~normalized the length of the cake to $1$, but the same proof works in our setting.}
\end{proposition}

\begin{proof}[Proof of \Cref{prop:GPAV-cake-EJR}]
For a positive integer $r$, one can check that the derivative with respect to $x$ of $\sum_{k=1}^r\frac{x}{k(x+k)}$ is $\sum_{k=1}^r\frac{1}{(x+k)^2}$, which converges as $r\rightarrow\infty$.
This means that $H_x = \sum_{k=1}^\infty \frac{x}{k(x+k)}$ is differentiable as a function of~$x$, and its derivative is $\sum_{k=1}^\infty\frac{1}{(x+k)^2}$.
In particular, there is no constant $c$ such that $H'_x = c/x$ for all $x\in (0, \infty)$---for example, this can be seen by observing that, as $x$ approaches $0$ from above, $H'_x$ approaches $\sum_{k = 1}^\infty 1/k^2 = \pi^2/6$ rather than $\infty$.
By \Cref{prop:maximize-EJR-M}, Generalized PAV does not satisfy cake EJR.
\end{proof}

\section{Proportionality Degree}\label{sec:proportionality-degree}

In addition to the axiomatic study of representation in terms of criteria like \EJRM and \EJROne, another relevant concept for cohesive groups is the \emph{proportionality degree}, which measures the average utility of the agents in each such group \citep{Skowron21}.
In this section, we first derive tight bounds on the proportionality degree implied by \EJRM and \EJROne, and then investigate the proportionality degree of the rules that we studied in \Cref{sec:EJR-notions}.

\begin{definition}[Average satisfaction]
Given an instance and an allocation~$A$, the \emph{average satisfaction} of a group of agents~$N' \subseteq N$ with respect to~$A$ is $\frac{1}{|N'|} \cdot \sum_{i \in N'} u_i(A)$.
\end{definition}

\begin{definition}[Proportionality degree]
Fix a function $f \colon \mathbb{R}_{> 0} \to \mathbb{R}_{\ge 0}$.
A rule~$\mathcal{M}$ has a \emph{proportionality degree} of~$f$ if for each instance~$I$, each allocation~$A$ that $\mathcal{M}$ outputs on~$I$, and each $t$-cohesive group of agents~$N^*$, the average satisfaction of~$N^*$ with respect to $A$ is at least~$f(t)$, i.e.,
\[
\frac{1}{|N^*|} \cdot \sum_{i \in N^*} u_i(A) \geq f(t).
\]
\end{definition}
For indivisible goods, \citet{SanchezFernandezElLa17} showed that EJR implies a proportionality degree of $\frac{t-1}{2}$.
We will show that in our setting, both \EJRM and \mbox{\EJROne} imply a proportionality degree of roughly $t/2$, with the guarantee for EJR-M being slightly higher.
In addition, we will establish that both GreedyEJR-M and Generalized MES have a proportionality degree of approximately $t/2$, while the proportionality degree of Generalized PAV is higher than $t-1$.

\subsection{Proportionality Degree Implied by EJR-M and EJR-1}

Our focus in this subsection is to establish tight bounds on the proportionality degree implied by \EJRM and \EJROne.
Observe that for $t < 1$, a $t$-cohesive group may have an average satisfaction of $0$ in an EJR-M or EJR-1 allocation.
Indeed, if $\alpha = t$ and the resource consists only of a single indivisible good, which is approved by all $n$ agents, then the set of all agents is $t$-cohesive, but the empty allocation is EJR-M and EJR-1.
We therefore assume $t\ge 1$ for our results from here on.

We first show that the proportionality degree implied by EJR-M is $\floor{t} \cdot \left( 1 - \frac{\floor{t} + 1}{2t} \right)$, beginning with the lower bound.
Note that this quantity is roughly $t/2$.

\begin{theorem}\label{thm:EJR-M-LB-average-satisfaction}
Given any instance and any real number \mbox{$t\ge 1$}, let $N^* \subseteq N$ be a $t$-cohesive group and $A$ be an \mbox{EJR-M} allocation.
The average satisfaction of $N^*$ with respect to $A$ is at least $\floor{t} \cdot \left( 1 - \frac{\floor{t} + 1}{2t} \right)$.
\end{theorem}

The high-level idea behind the proof of \Cref{thm:EJR-M-LB-average-satisfaction} is that, given a $t$-cohesive group~$N^*$ and an \EJRM allocation, a $\frac{t - \floor{t}}{t}$ fraction of the agents in $N^*$ are guaranteed a utility of at least~$\floor{t}$.
The remaining agents can then be partitioned into $\floor{t}$ disjoint subsets so that each subset consists of a $1/t$ fraction of the agents in $N^*$ and the guaranteed utilities for these subsets drop arithmetically from~$\floor{t}-1$ to~$0$.

\begin{proof}[Proof of \Cref{thm:EJR-M-LB-average-satisfaction}]
For ease of notation, let~$r \coloneqq n/\alpha$, and note that $|N^*| \ge \ceiling{tr}$.
Since $N^*$ is $t$-cohesive, by \Cref{prop:EJR-M-floor}, some agent~$i_1 \in N^*$ gets utility at least~$\floor{t}$ from the allocation~$A$.
If $|N^* \setminus \{i_1\}| \geq \floor{t} \cdot r$, then since $N^* \setminus \{i_1\}$ is $\floor{t}$-cohesive, \Cref{prop:EJR-M-floor} implies that another agent~$i_2 \neq i_1$ gets utility at least~$\floor{t}$ from~$A$.
Applying this argument repeatedly, as long as there are at least $\ceiling*{\floor{t} \cdot r}$~agents left, \Cref{prop:EJR-M-floor} implies that one of them gets utility at least~$\floor{t}$.
Let $N'_{\floor{t}}$ consist of the agents with guaranteed utility $\floor{t}$ from this argument, and note that $|N'_{\floor{t}}| = |N^*| - \ceiling{\floor{t} \cdot r} + 1 \ge \ceiling{tr} - \ceiling{\floor{t} \cdot r} + 1$.
Let $\widehat{N} \coloneqq N^* \setminus N'_{\floor{t}}$; we have $|\widehat{N}| = \ceiling{\floor{t} \cdot r} - 1$.
Denote by~$N_{\floor{t}}$ an arbitrary subset of~$N'_{\floor{t}}$ of size exactly $\ceiling{tr} - \ceiling{\floor{t} \cdot r} + 1$.

Now, let us consider the agents in~$\widehat{N}$.
Applying an argument similar to the one in the previous paragraph but using $(\floor{t}-1)$-cohesiveness, we find that $\widehat{N}$ contains at least $\ceiling*{\floor{t} \cdot r} - \ceiling*{(\floor{t} - 1) \cdot r}$ agents with a utility of at least~$\floor{t} - 1$ each; let these agents form~$N_{\floor{t}-1}$.
Continuing inductively, we can partition $\widehat{N}$ into~$\floor{t}$ pairwise disjoint sets $N_{\floor{t} - 1}, N_{\floor{t} - 2}, \dots, N_1, N_0$ such that for each~$j \in \{0, 1, \dots, \floor{t} - 1\}$, every agent in~$N_j$ gets utility at least~$j$ from the allocation~$A$.

For each~$j \in \{1, 2, \dots, \floor{t}\}$, it holds that $\left| \bigcup_{k = 0}^{j-1} N_k \right| = \ceiling{jr} - 1$.
Furthermore, we have
\begin{align*}
j \cdot \ceiling{tr} &\geq j \cdot t \cdot r
= t \cdot (jr + 1) - t \geq t \cdot \ceiling{jr} - t = t \cdot (\ceiling{jr} - 1),
\end{align*}
which implies that
\[
\frac{\left| \bigcup_{k = 0}^{j-1} N_k \right|}{\ceiling{tr}} = \frac{\ceiling{jr} - 1}{\ceiling{tr}} \leq \frac{j}{t}.
\]
Since $\left| \bigcup_{k = 0}^{\floor{t}} N_k \right| = \ceiling{tr}$, it follows that
\begin{align}\label{eq:EJR-M-agents-fraction}
\frac{\left| \bigcup_{k = j}^{\floor{t}} N_k \right|}{\ceiling{tr}}
\geq \frac{t - j}{t}
&= \frac{t - \floor{t}}{t} + \frac{\floor{t} - j}{t}
= \frac{t - \floor{t}}{t} + \sum_{k = j}^{\floor{t}-1} \frac{1}{t}.
\end{align}
With this relationship in hand, we can bound the average satisfaction of $N_{\floor{t}} \cup \widehat{N} = \bigcup_{k = 0}^{\floor{t}} N_k$ as
\begin{align*}
\frac{1}{\left| \bigcup_{k = 0}^{\floor{t}} N_k \right|} \cdot \sum_{i \in \bigcup_{k = 0}^{\floor{t}} N_k} u_i(A)
&\geq \frac{1}{\ceiling{tr}} \cdot \left( \sum_{k = 0}^{\floor{t}} \left| N_k \right| \cdot k \right) \\
&= \sum_{k = 0}^{\floor{t}} \frac{\left| N_k \right|}{\ceiling{tr}} \cdot k \\
&= \sum_{d = 1}^{\floor{t}} \sum_{k = d}^{\floor{t}} \frac{\left| N_k \right|}{\ceiling{tr}} \\
&\geq \sum_{d = 1}^{\floor{t}} \left( \frac{t - \floor{t}}{t} + \sum_{k = d}^{\floor{t}-1} \frac{1}{t} \right)  \\
&= \frac{t - \floor{t}}{t} \cdot \floor{t} + \sum_{d = 1}^{\floor{t}} \sum_{k = d}^{\floor{t}-1} \frac{1}{t} \\
&= \frac{t - \floor{t}}{t} \cdot \floor{t} + \frac{1}{t} \cdot \frac{\floor{t} \cdot (\floor{t} - 1)}{2} \\
&= \frac{\floor{t}}{t} \cdot \frac{2t - \floor{t} - 1}{2} \\
&= \floor{t} \cdot \left( 1 - \frac{\floor{t}+1}{2t} \right),
\end{align*}
where the first inequality holds because each agent in $N_k$ gets utility at least $k$ and the second inequality follows from~\eqref{eq:EJR-M-agents-fraction}.

Since every agent in~$N'_{\floor{t}} \setminus N_{\floor{t}}$ gets utility at least $\floor{t}$, the average satisfaction of $N'_{\floor{t}} \setminus N_{\floor{t}}$ is at least~$\floor{t} \ge \floor{t} \cdot \left( 1 - \frac{\floor{t}+1}{2t} \right)$.
As the average satisfaction of~$N^*$ is a convex combination of the corresponding quantities for $N'_{\floor{t}} \setminus N_{\floor{t}}$ and $N_{\floor{t}} \cup \widehat{N}$, it is at least~$\floor{t} \cdot \left( 1 - \frac{\floor{t}+1}{2t} \right)$, as desired.
\end{proof}

We next give a matching upper bound.

\begin{theorem}\label{thm:EJR-M-UB-average-satisfaction}
For any real numbers $t \geq 1$ and~$\varepsilon > 0$, there exists an instance, a $t$-cohesive group~$N^*$, and an \mbox{EJR-M} allocation~$A$ such that the average satisfaction of $N^*$ with respect to $A$ is at most $\floor{t} \cdot \left( 1 - \frac{\floor{t}+1}{2t} \right) + \varepsilon$.
\end{theorem}

We do not prove \Cref{thm:EJR-M-UB-average-satisfaction} directly, as we will establish a stronger statement later in \Cref{thm:prop-degree-GreedyEJR-M}.

Next, we show that the proportionality degree implied by \EJROne is $\frac{t - 2 + 1/t}{2} = \frac{(t-1)^2}{2t}$, which is slightly lower than that implied by \EJRM for every $t > 1$.
For the lower bound, we use a similar idea as in \Cref{thm:EJR-M-LB-average-satisfaction}, but we need to be more careful about agents with low utility guarantees.
In particular, even when the guarantee provided by the EJR-1 condition is negative, the actual utility is always nonnegative, so we need to ``round up'' the EJR-1 guarantee appropriately.

\begin{theorem}\label{thm:EJR-1-LB-average-satisfaction}
Given any instance and any real number \mbox{$t\ge 1$}, let $N^* \subseteq N$ be a $t$-cohesive group and $A$ be an \mbox{EJR-1} allocation.
The average satisfaction of $N^*$ with respect to $A$ is greater than $\frac{t-2+1/t}{2}$.
\end{theorem}

To prove this \namecref{thm:EJR-1-LB-average-satisfaction}, we will use the following claim, which provides a lower bound for the average of a nonincreasing and nonnegative sequence with a particular structure.

\begin{claim}\label{claim:average}
Let~$r > 0$ and~$t \geq 1$ be real numbers.
Consider any nonincreasing and nonnegative sequence
\[
t-1, a_1, a_2, \dots, a_{\ceiling{tr} - \ceiling{r}}, b_1, b_2, \dots, b_{\ceiling{r}-1},
\]
in which $a_1, a_2, \dots, a_{\ceiling{tr} - \ceiling{r}}$ forms an arithmetic subsequence with common difference $-1/r$.
If $t-1 - a_1 \leq 1/r$, then the average of the entire sequence is at least $\frac{t-2+1/t}{2}$.
\end{claim}

\begin{proof}
We start by showing that the average of the subsequence $t-1, a_1, a_2, \dots, a_{\ceiling{tr} - \ceiling{r}}$ is at least $\frac{t-1}{2}$.
The bound holds trivially if~$\ceiling{tr} - \ceiling{r} = 0$; we therefore assume that $\ceiling{tr} - \ceiling{r} \geq 1$.
Let us continually decrease each of the numbers $a_1, a_2, \dots, a_{\ceiling{tr} - \ceiling{r}}$ by the same amount until (at least) one of the following two cases occurs:
\begin{itemize}
\item \underline{Case~1}:
The difference between~$t-1$ and~$a_1$ becomes~$1/r$, i.e., $a_1 = t-1 - 1/r$.
Note that $a_{\ceiling{tr}-\ceiling{r}}$ is still nonnegative in this case.

\item \underline{Case~2}:
$a_{\ceiling{tr} - \ceiling{r}}$ becomes~$0$.
Note that the difference between~$t-1$ and~$a_1$ is still at most~$1/r$.
\end{itemize}
Clearly, the average of the subsequence in question $t-1, a_1, a_2, \dots, a_{\ceiling{tr} - \ceiling{r}}$ does not increase during this process.
Thus, it suffices to show that in each of the above two cases, this average is at least~$\frac{t-1}{2}$ after the process.

\begin{itemize}
\item In Case~1, the subsequence $t-1, a_1, a_2, \dots, a_{\ceiling{tr} - \ceiling{r}}$ is now an arithmetic sequence, so its average is $\frac{(t-1) + a_{\ceiling{tr} - \ceiling{r}}}{2} \geq \frac{t-1}{2}$, where the inequality follows from the fact that $a_{\ceiling{tr} - \ceiling{r}} \geq 0$ in this case.

\item In Case~2, consider the arithmetic sequence~$(d_k)_{k = 0}^{\ceiling{tr} - \ceiling{r}}$ with $d_0 = t-1$ and $d_{\ceiling{tr} - \ceiling{r}} = 0$.
Let~$-\beta$ be its common difference, so~$\beta$ is nonnegative.
On the one hand, we have
\[
t - 1 = d_0 = d_{\ceiling{tr} - \ceiling{r}} + \beta \cdot (\ceiling{tr} - \ceiling{r}) = \beta \cdot (\ceiling{tr} - \ceiling{r}).
\]
On the other hand, we have
\begin{align*}
t-1 &= (t-1 - a_1) + a_1 \\
&= (t-1 - a_1) + (\ceiling{tr} - \ceiling{r} - 1) \cdot 1/r
\leq (\ceiling{tr} - \ceiling{r}) \cdot 1/r,
\end{align*}
where the inequality holds because $t-1 - a_1 \leq 1/r$ in Case~2.
As a result, we have $\beta \cdot (\ceiling{tr} - \ceiling{r}) \leq (\ceiling{tr} - \ceiling{r}) \cdot 1/r$, that is, $\beta \leq 1/r$.
Hence, each term of the sequence $t-1, a_1, a_2, \dots, a_{\ceiling{tr} - \ceiling{r}}$ is at least as large as the corresponding term of the sequence~$(d_k)_{k = 0}^{\ceiling{tr} - \ceiling{r}}$.
We conclude that the average of $t-1, a_1, a_2, \dots, a_{\ceiling{tr} - \ceiling{r}}$ is at least that of $(d_k)_{k = 0}^{\ceiling{tr} - \ceiling{r}}$, which is $\frac{d_0 + d_{\ceiling{tr} - \ceiling{r}}}{2} = \frac{t-1}{2}$.
\end{itemize}
In both cases, we have proven that the average of the sequence $t-1, a_1, a_2, \dots, a_{\ceiling{tr} - \ceiling{r}}$ is at least $\frac{t-1}{2}$.

Next, we show that the average of the entire sequence
\[t-1, a_1, a_2, \dots, a_{\ceiling{tr} - \ceiling{r}}, b_1, b_2, \dots, b_{\ceiling{r}-1}\]
is at least~$\frac{t-2+1/t}{2}$.
This can be done by taking all $b_i$'s to be $0$ and applying the lower bound on the average of $t-1, a_1, a_2, \dots, a_{\ceiling{tr} - \ceiling{r}}$ that we previously computed:
\begin{align*}
\frac{1}{1 + (\ceiling{tr}-\ceiling{r}) + (\ceiling{r}-1)} \cdot \frac{t-1}{2} \cdot (1 + (\ceiling{tr} - \ceiling{r}))
&= \frac{1 + \ceiling{tr} - \ceiling{r}}{\ceiling{tr}} \cdot \frac{t-1}{2} \\
&= \left( 1 + \frac{1 - \ceiling{r}}{\ceiling{tr}} \right) \cdot \frac{t-1}{2} \\
&\geq \left( 1 + \frac{1 - (r+1)}{\ceiling{tr}} \right) \cdot \frac{t-1}{2} \\
&= \left( 1 - \frac{r}{\ceiling{tr}} \right) \cdot \frac{t-1}{2} \\
&\geq \left( 1 - \frac{r}{tr} \right) \cdot \frac{t-1}{2} \\
&= \frac{(t-1)^2}{2t} \\
&= \frac{t-2+1/t}{2}.
\end{align*}
The claim is thus proven.
\end{proof}

We are now ready to establish \Cref{thm:EJR-1-LB-average-satisfaction}.

\begin{proof}[Proof of \Cref{thm:EJR-1-LB-average-satisfaction}]
For notational convenience, let~$r \coloneqq n/\alpha$.
We have $|N^*| \geq t r$, where the inequality holds because $N^*$ is $t$-cohesive.
\EJROne implies that some agent~$i_1 \in N^*$ gets utility greater than~$t-1$ from the allocation~$A$.
If $|N^* \setminus \{i_1\}| \geq t r$, then since there still exists a subset of the resource of size at least $t$ commonly approved by the agents in $N^* \setminus \{i_1\}$, \EJROne implies that another agent~$i_2 \neq i_1$ gets utility greater than~$t-1$ from~$A$.
Applying this argument repeatedly, as long as there are at least~$t \cdot r$ agents left, EJR-1 implies that one of them gets utility greater than~$t-1$.
Let $N_{t-1}$ consist of the agents with guaranteed utility greater than $t-1$ from this argument.
Let $\widehat{N} \coloneqq N^* \setminus N_{t-1}$ and $\widehat{n} \coloneqq |\widehat{N}|$, and note that $\widehat{n} = \ceiling{tr} - 1$.

Now, let us consider the agents in~$\widehat{N}$.
Since $|\widehat{N}| = \widehat{n} \ge \frac{\widehat{n}}{r} \cdot r$ and $s\left( \bigcap_{i \in \widehat{N}} R_i \right) \geq t > \frac{\widehat{n}}{r}$, the agents in~$\widehat{N}$ form an $\frac{\widehat{n}}{r}$-cohesive group.
By \EJROne, some agent in~$\widehat{N}$ gets utility greater than~$\frac{\widehat{n}}{r} - 1$.
Continuing inductively, the guaranteed utility drops arithmetically with a common difference of~$1/r$.
Note also that an agent's actual utility is always nonnegative.

To calculate the average satisfaction of the $t$-cohesive group~$N^*$, we first focus on the agents in $\widehat{N}$ along with agent~$i_1$ discussed earlier in the proof.
The average satisfaction of these agents is
\begin{align*}
\frac{1}{1 + \widehat{n}} \cdot \left( u_{i_1}(A) + \sum_{i \in \widehat{N}} u_i(A) \right)
&> \frac{1}{1 + \widehat{n}} \cdot \left( t-1 + \sum_{i \in \widehat{N}} u_i(A) \right) \\
&\ge \frac{1}{1 + \widehat{n}} \cdot \left( t-1 + \sum_{j = \ceiling{r}}^{\widehat{n}} \left( \frac{j}{r} - 1 \right) + \sum_{j = 1}^{\ceiling{r} - 1} 0 \right) \\
&\geq \frac{t-2+1/t}{2}.
\end{align*}
The last inequality is due to \Cref{claim:average}: We have a nonincreasing and nonnegative sequence with $t-1$ as the first element, followed by a decreasing arithmetic sequence with $\ceiling{tr} - \ceiling{r}$ terms whose common difference is~$-1/r$ and whose first term, $\frac{\ceiling{tr}-1}{r} - 1$, is at most~$1/r$ away from~$t-1$, and then followed by $\ceiling{r} - 1$ zeros.
The average satisfaction of~$N_{t-1} \setminus \{i_1\}$, if this set is not empty, is greater than~$t-1$.
As the average satisfaction of~$N^*$ is a convex combination of the corresponding quantities for $N_{t-1} \setminus \{i_1\}$ and $\widehat{N} \cup \{i_1\}$, it is greater than~$\frac{t-2+1/t}{2}$, as desired.
\end{proof}

We now derive a matching upper bound.

\begin{theorem}\label{thm:EJR-1-UB-average-satisfaction}
For any real numbers $t \geq 1$ and~$\varepsilon > 0$, there exists an instance, a $t$-cohesive group~$N^*$, and an \mbox{EJR-1} allocation~$A$ such that the average satisfaction of $N^*$ with respect to $A$ is at most $\frac{t-2+1/t}{2} + \varepsilon$.
\end{theorem}

\begin{proof}
Consider a cake instance with a sufficiently large number of agents~$n$ (to be specified later).
Let~$\alpha = t$.
Thus, we have $|N| = t \cdot n/t \geq t \cdot n/\alpha$.
The cake is given by the interval $[0, 2t]$, and the agents' preferences are as follows.
\begin{itemize}
\item Each agent~$i \in \left\{ 1, 2, \dots, \ceiling*{n/\alpha}-1 \right\}$ approves the interval~$[0, t]$.
\item Each agent~$i \in \left\{ \ceiling*{n/\alpha}, \ceiling*{n/\alpha}+1, \dots, n \right\}$ approves the interval~$\left[ 0, t + \frac{i - n/\alpha}{n/\alpha} + \delta \right]$, where~$\delta \in (0,1)$ is sufficiently small (to be specified later).
\end{itemize}
Since all $n$ agents approve the interval~$[0, t]$, they form a $t$-cohesive group~$N$.

We claim that allocation~$A = [t, 2t]$, which has size $t = \alpha$, satisfies \EJROne.
Consider a $t'$-cohesive group for some value of $t' > 0$.
If~$t' \in (0, 1)$, the requirement of EJR-1 is trivially fulfilled.
Since $n = t\cdot n/\alpha$, we may therefore assume that $t'\in [1, t]$.
Consider agent~$\ceiling*{t' \cdot n/\alpha} \in N$, who approves the interval~$\left[ 0, t + \frac{\ceiling*{t' \cdot n/\alpha} - n/\alpha}{n/\alpha} + \delta \right]$; this agent gets utility at least $\frac{\ceiling*{t' \cdot n/\alpha} - n/\alpha}{n/\alpha} + \delta \geq \frac{t' \cdot n/\alpha - n/\alpha}{n/\alpha} + \delta > t' - 1$ from the allocation~$A$.
Since every $t'$-cohesive group contains at least~$\ceiling*{t' \cdot n/\alpha}$ agents, it must contain an agent who gets utility greater than~$t' - 1$ from~$A$.
This means that $A$ satisfies \EJROne, as claimed.

The average satisfaction of the $t$-cohesive group~$N$ with respect to the \EJROne allocation~$A$ is
\begin{align*}
\frac{1}{|N|} \cdot \sum_{i \in N} u_i(A)
&= \frac{1}{n} \cdot \sum_{i = \ceiling{n/\alpha}}^{n} \left( \frac{i - n/\alpha}{n/\alpha} + \delta \right) \\
&= \frac{\alpha}{n^2} \cdot \sum_{i = \ceiling{n/\alpha}}^{n} (i - n/\alpha) + \frac{1}{n} \cdot \sum_{i = \ceiling{n/\alpha}}^{n} \delta \\
&= \frac{\alpha}{n^2} \cdot \frac{(\ceiling{n/\alpha} - n/\alpha) + (n - n/\alpha)}{2} \cdot (n - \ceiling{n/\alpha} + 1) \\
&\quad+ \frac{\delta}{n} \cdot (n - \ceiling{n/\alpha} + 1) \\
&\leq \frac{\alpha}{n^2} \cdot \frac{(n - n/\alpha + 1)^2}{2} + \frac{\delta}{n} \cdot (n - n/\alpha + 1) \\
&= \frac{\alpha - 2 + 1/\alpha}{2} + \frac{\delta \cdot (\alpha - 1)}{\alpha} + \frac{\alpha}{2n^2} + \frac{\alpha - 1 + \delta}{n} \\
&= \frac{t - 2 + 1/t}{2} + \frac{\delta \cdot (t-1)}{t} + \frac{t}{2n^2} + \frac{t-1+\delta}{n},
\end{align*}
where the inequality holds because $\ceiling{n/\alpha} - n/\alpha \le 1$ and $-\ceiling{n/\alpha} \le -n/\alpha$.
Finally, we choose a sufficiently large~$n$ and a sufficiently small~$\delta$ so that $\frac{\delta \cdot (t-1)}{t} + \frac{t}{2n^2} + \frac{t-1+\delta}{n} \leq \varepsilon$; this ensures that the average satisfaction of $N$ is at most $\frac{t -2+1/t}{2} + \varepsilon$, as desired.
\end{proof}

\subsection{Proportionality Degree of Specific Rules}

In this subsection, we investigate the proportionality degree of the rules that we studied in \Cref{sec:EJR-notions}.

We begin with GreedyEJR-M.
Since GreedyEJR-M satisfies \EJRM, \Cref{thm:EJR-M-LB-average-satisfaction} immediately yields a lower bound.
We derive a matching upper bound, which implies that the proportionality degree of GreedyEJR-M is $\floor{t} \cdot \left( 1 - \frac{\floor{t} + 1}{2 t} \right)$.

\begin{theorem}\label{thm:prop-degree-GreedyEJR-M}
For any real numbers $t \geq 1$ and $\varepsilon > 0$, there exists an instance, a $t$-cohesive group $N^*$, and an allocation $A$ output by GreedyEJR-M such that the average satisfaction of $N^*$ with respect to $A$ is at most $\floor{t} \cdot \left( 1 - \frac{\floor{t} + 1}{2t} \right) + \varepsilon$.
\end{theorem}

We first provide an intuition behind the proof of \Cref{thm:prop-degree-GreedyEJR-M}.
We construct an indivisible-goods instance, make $\alpha$ an integer, and choose $n$ to be a multiple of~$\alpha$.
Our goal is to construct a target $t$-cohesive group of agents~$N^*$ with as small utilities as possible.
Since GreedyEJR-M outputs an \EJRM allocation, the largest number of agents in $N^*$ that receive utility $0$---denote the set of these agents by~$N_0$---is $n/\alpha - 1$; otherwise, these agents would form a $1$-cohesive group and cannot all receive utility~$0$.
Similarly, among the agents in~$N^* \setminus N_0$, the largest number of agents that receive utility~$1$---denote the set of these agents by $N_1$---is $n/\alpha$, as we do not want $N_0\cup N_1$ to form a $2$-cohesive group.
Continuing inductively, we want to partition $N^*$ into $N_0\cup N_1\cup\dots\cup N_{\floor{t}}$, with the agents in $N_k$ receiving utility exactly~$k$ for each~$k$.
We add dummy agents and goods in order to make sure that, instead of all agents in $N^*$ being satisfied at once by the GreedyEJR-M execution, the agents in $N_{\floor{t}}$ are first satisfied along with some dummy agents via some dummy goods, then those in $N_{\floor{t}-1}$ are satisfied along with other dummy agents via other dummy goods, and so on.
The dummy agents and goods need to be carefully constructed to make this argument work.

\begin{proof}[Proof of \Cref{thm:prop-degree-GreedyEJR-M}]
Let~$\alpha = \frac{\floor{t} \cdot (\floor{t}+1)}{2} + 1 = \frac{\floor{t}^2 + \floor{t} + 2}{2}$, and note that $\alpha$ is an integer.
We will construct an indivisible-goods instance with a sufficiently large number of agents~$n$ (to be specified later), where $n$ is a multiple of~$\alpha$ that is at least $2\alpha$.
Observe that
\begin{align*}
\ceiling*{t \cdot \frac{n}{\alpha}} &= \ceiling*{\floor{t} \cdot \frac{n}{\alpha} + (t-\floor{t}) \cdot \frac{n}{\alpha}}
= \floor{t} \cdot \frac{n}{\alpha} + \ceiling*{(t-\floor{t}) \cdot \frac{n}{\alpha}}.
\end{align*}
Since $t-\floor{t} < 1$, we can choose $n$ large enough so that  $\ceiling*{(t-\floor{t}) \cdot \frac{n}{\alpha}} \leq \frac{n}{\alpha} - 1$.
When this holds, we have
\begin{align}
t \cdot \frac{n}{\alpha} &\leq \ceiling*{t \cdot \frac{n}{\alpha}}
\leq \floor{t} \cdot \frac{n}{\alpha} + \frac{n}{\alpha} - 1 = (\floor{t}+1) \cdot \frac{n}{\alpha} - 1.
\label{eq:squeeze}
\end{align}
This inequality will help ensure that we can construct a $t$-cohesive group that is not $(\floor{t}+1)$-cohesive.
Note that in an indivisible-goods instance, a $t$-cohesive group must commonly approve at least~$\ceiling{t}$ indivisible goods.

We now describe our instance.
Let~$G = \{g_1, g_2, \dots, g_{\ceiling{t}}\} \cup \left( \bigcup_{k = 1}^{\floor{t}} D^G_k \right)$ be the set of indivisible goods, where for each~$k \in \{1, 2, \dots, \floor{t}\}$, $D^G_k$ contains exactly~$k$ indivisible goods.
In particular, the sets
\[
\{g_1, g_2, \dots, g_{\ceiling{t}}\}, D^G_1, D^G_2, \dots, D^G_{\floor{t}}
\]
are all disjoint.
Our specifications for the agents are slightly different depending on whether $t \geq 2$ or $t\in [1,2)$; we distinguish between the two cases below.

\paragraph{Case~1: $t \geq 2$.}

Recalling that $n/\alpha$ is an integer, we partition the set of all agents~$N$ into the following pairwise disjoint sets:
\begin{align*}
&N_0 = \left\{ 1, 2, \dots, \frac{n}{\alpha}-1 \right\}, \\
&N_1 = \left\{ \frac{n}{\alpha}, \frac{n}{\alpha}+1, \dots, 2 \cdot \frac{n}{\alpha}-1 \right\}, \\
&\qquad\vdots \\
&N_k = \left\{ k \cdot \frac{n}{\alpha}, k \cdot \frac{n}{\alpha}+1, \dots, (k+1) \cdot \frac{n}{\alpha}-1 \right\}, \\
&\qquad\vdots \\
&N_{\floor{t} -1} = \left\{ (\floor{t}-1) \cdot \frac{n}{\alpha}, (\floor{t}-1) \cdot \frac{n}{\alpha}+1, \dots, \floor{t} \cdot \frac{n}{\alpha}-1 \right\}, \\
&N_{\floor{t}} = \left\{ \floor{t} \cdot \frac{n}{\alpha}, \floor{t} \cdot \frac{n}{\alpha}+1, \dots, \ceiling*{t \cdot \frac{n}{\alpha}} \right\}, \\
&D_1, D_2, \dots, D_{\floor{t}-1}, D_{\floor{t}},
\end{align*}
where $N^* \coloneqq \bigcup_{k = 0}^{\floor{t}} N_k = \left\{ 1, 2, \dots, \ceiling*{t \cdot \frac{n}{\alpha}} \right\}$ is our target $t$-cohesive group and $\bigcup_{k = 1}^{\floor{t}} D_k$ consists of ``dummy agents''.
More specifically:
\begin{itemize}
\item $D_1$ contains a single agent;
\item For each~$k \in \{2, \dots, \floor{t}-1\}$, $D_k$ contains $(k-1) \cdot \frac{n}{\alpha}$ agents;
\item $D_{\floor{t}}$ contains $\floor{t} \cdot \frac{n}{\alpha} - \left( \ceiling*{t \cdot \frac{n}{\alpha}} - \floor{t} \cdot \frac{n}{\alpha} + 1 \right) = 2 \floor{t} \cdot \frac{n}{\alpha} - \ceiling*{t \cdot \frac{n}{\alpha}} - 1$ agents.
Observe that $\left| N_{\floor{t}} \cup D_{\floor{t}} \right| = \floor{t} \cdot \frac{n}{\alpha}$.
\end{itemize}
Note that $D_1$ and $D_{\floor{t}}$ are different sets because $t\ge 2$.
We verify that the total number of agents is indeed $n$:
\begin{align*}
|N| &= \left| \left( \bigcup_{k = 0}^{\floor{t}} N_k \right) \cup \left( \bigcup_{k = 1}^{\floor{t}} D_k \right) \right| \\
&= \left| \left( \bigcup_{k = 0}^{\floor{t}-1} N_k \right) \cup D_1 \cup \left( \bigcup_{k = 2}^{\floor{t}-1} D_k \right) \cup \left( N_{\floor{t}} \cup D_{\floor{t}} \right) \right| \\
&= \left( \floor{t} \cdot \frac{n}{\alpha} - 1 \right) + 1 + \sum_{k = 2}^{\floor{t}-1} (k-1) \cdot \frac{n}{\alpha} + \floor{t} \cdot \frac{n}{\alpha} \\
&= 2 \floor{t} \cdot \frac{n}{\alpha} + \frac{n}{\alpha} \cdot \frac{(2 - 1) + ((\floor{t}-1) - 1)}{2}\cdot ((\floor{t}-1) - 2 + 1) \\
&= \frac{n}{\alpha} \cdot \left( 2 \floor{t} + \frac{(\floor{t}-1) \cdot (\floor{t}-2)}{2} \right) \\
&= \frac{n}{\alpha} \cdot \frac{\floor{t}^2 + \floor{t} + 2}{2} \\
&= n.
\end{align*}
The agents' preferences are as follows.
\begin{itemize}
\item The agents in~$N_0$ approve the goods in~$\{g_1, g_2, \dots, g_{\ceiling{t}}\}$.
\item For each~$k \in \{1, 2, \dots, \floor{t}\}$, the agents in~$N_k$ approve the goods in~$\{g_1, g_2, \dots, g_{\ceiling{t}}\} \cup D^G_k$, and the agents in~$D_k$ approve the goods in~$D^G_k$.
\end{itemize}
This completes the description of our instance.
Since $|N^*| = \ceiling*{t \cdot \frac{n}{\alpha}}$, inequality \eqref{eq:squeeze} implies that $N^*$ is not $(\floor{t}+1)$-cohesive.
On the other hand, since the agents in~$N^*$ commonly approve the goods~$g_1, g_2, \dots, g_{\ceiling{t}}$, $N^*$ is $t$-cohesive.
Also, notice that $t^* = \floor{t}$ is the largest integer such that a $t^*$-cohesive group exists in our instance.

We now consider the execution of GreedyEJR-M on the above instance.
Since our instance consists exclusively of indivisible goods, only integers~$t^*$ are relevant for the \mbox{\EJRM} condition in each round of GreedyEJR-M.
We claim that GreedyEJR-M can return the allocation~$A = \bigcup_{k = 1}^{\floor{t}} D^G_k$, which has size~$\frac{\floor{t} \cdot (\floor{t}+1)}{2} \leq \alpha$.
Given the instance, at the beginning, $t^* = \floor{t}$ is the largest number such that the \EJRM condition is satisfied: it is satisfied with $\left( N_{\floor{t}} \cup D_{\floor{t}}, D^G_{\floor{t}} \right)$, because $\left| N_{\floor{t}} \cup D_{\floor{t}} \right| = \floor{t} \cdot \frac{n}{\alpha}$ and the agents in $N_{\floor{t}} \cup D_{\floor{t}}$ commonly approve all goods in~$D^G_{\floor{t}}$, which has size exactly~$\floor{t}$.\footnote{At this stage, the \EJRM condition with $t^* = \floor{t}$ also holds with~$(N^*, \{g_1, g_2, \dots, g_{\floor{t}}\})$.}
\mbox{GreedyEJR-M} removes the agents in~$N_{\floor{t}} \cup D_{\floor{t}}$ and adds the goods in~$D^G_{\floor{t}}$ to the allocation~$A$.
Now, there is no more $\floor{t}$-cohesive group, and $t^* = \floor{t}-1$ becomes the largest number such that the \EJRM condition is satisfied: it is satisfied with~$\left( N_{\floor{t}-1} \cup D_{\floor{t}-1}, D^G_{\floor{t}-1} \right)$.
More generally, for each integer~$k$ from~$\floor{t}-1$ down to~$2$, the \EJRM condition is satisfied for~$t^* = k$ with~$\left( N_k \cup D_k, D^G_k \right)$ because
\begin{align*}
|N_k \cup D_k|
&= \frac{n}{\alpha} + (k-1) \cdot \frac{n}{\alpha}
= k \cdot \frac{n}{\alpha}
\end{align*}
and the agents in~$N_k \cup D_k$ commonly approve the goods in~$D^G_k$, which has size exactly~$k$; moreover,
\begin{align*}
|N_0\cup N_1\cup\dots\cup N_k| &= (k+1)\cdot\frac{n}{\alpha} - 1
< (k+1)\cdot\frac{n}{\alpha}.
\end{align*}
Hence, at each stage, GreedyEJR-M can remove the agents in~$N_k \cup D_k$ and add the goods in~$D^G_k$ to the allocation~$A$.
Finally, when only the agents in~$N_0 \cup N_1 \cup D_1$ remain, there is no $2$-cohesive group; however, the \EJRM condition is satisfied for~$t^* = 1$ with~$\left( N_1 \cup D_1, D^G_1 \right)$ because $|N_1 \cup D_1| = \frac{n}{\alpha} + 1 \geq \frac{n}{\alpha}$ and the agents in~$N_1 \cup D_1$ commonly approve the single good in~$D^G_1$.
Once $N_1 \cup D_1$ is removed by GreedyEJR-M and $D^G_1$ is added to~$A$, the remaining agents in~$N_0$ do not form a $1$-cohesive group, so \mbox{GreedyEJR-M} terminates.
Note that for every~$k \in \{0, 1, 2, \dots, \floor{t}\}$, each agent in~$N_k$ gets utility exactly~$k$ from the allocation~$A$.

The average satisfaction of $N^*$ with respect to the returned allocation $A$ is
\begin{align*}
\frac{\sum_{i \in N^*} u_i(A)}{|N^*|}
&= \frac{1}{|N^*|} \cdot \sum_{k = 0}^{\floor{t}} |N_k| \cdot k \\
&= \sum_{k = 0}^{\floor{t}-1} \frac{|N_k| \cdot k}{|N^*|} + \frac{\left| N_{\floor{t}} \right| \cdot \floor{t}}{|N^*|} \\
&= \sum_{d = 1}^{\floor{t}-1} \sum_{k = d}^{\floor{t}-1} \frac{|N_k|}{|N^*|} + \frac{\left| N_{\floor{t}} \right| \cdot \floor{t}}{|N^*|} \\
&= \sum_{d = 1}^{\floor{t}-1} \frac{\left| \bigcup_{k = d}^{\floor{t}-1} N_k \right|}{|N^*|} + \frac{\left| N_{\floor{t}} \right| \cdot \floor{t}}{|N^*|} \\
&= \sum_{d = 1}^{\floor{t}-1} \frac{\left( \floor{t} \cdot \frac{n}{\alpha} - 1 \right) - \left( d \cdot \frac{n}{\alpha} - 1 \right)}{\ceiling*{t \cdot \frac{n}{\alpha}}}+ \frac{\left( \ceiling*{t \cdot \frac{n}{\alpha}} - \floor{t} \cdot \frac{n}{\alpha} + 1 \right) \cdot \floor{t}}{\ceiling*{t \cdot \frac{n}{\alpha}}}  \\
&= \sum_{d = 1}^{\floor{t}-1} \frac{(\floor{t} - d) \cdot \frac{n}{\alpha}}{\ceiling*{t \cdot \frac{n}{\alpha}}}
+ \frac{\left( \ceiling*{\floor{t} \cdot \frac{n}{\alpha} + (t - \floor{t}) \cdot \frac{n}{\alpha}} - \floor{t} \cdot \frac{n}{\alpha} + 1 \right) \cdot \floor{t}}{\ceiling*{t \cdot \frac{n}{\alpha}}} \\
&\leq \sum_{d = 1}^{\floor{t}-1} \frac{(\floor{t} - d) \cdot \frac{n}{\alpha}}{t \cdot \frac{n}{\alpha}}
+ \frac{\left( \floor{t} \cdot \frac{n}{\alpha} + \ceiling*{(t - \floor{t}) \cdot \frac{n}{\alpha}} - \floor{t} \cdot \frac{n}{\alpha} + 1 \right) \cdot \floor{t}}{t \cdot \frac{n}{\alpha}} \\
&\leq \frac{1}{t} \cdot \sum_{d = 1}^{\floor{t}-1} (\floor{t} - d) + \frac{\left( (t - \floor{t}) \cdot \frac{n}{\alpha} + 2 \right) \cdot \floor{t}}{t \cdot \frac{n}{\alpha}} \\
&= \frac{1}{t} \cdot \frac{\floor{t} \cdot (\floor{t} - 1)}{2} + \frac{\floor{t} \cdot (t - \floor{t})}{t} + \frac{2 \alpha \floor{t}}{n t} \\
&= \floor{t} \cdot \left( \frac{\floor{t} - 1}{2t} + \frac{t - \floor{t}}{t} \right) + \frac{\floor{t} \cdot (\floor{t}^2 + \floor{t} + 2)}{n t} \\
&= \floor{t} \cdot \left( 1 - \frac{\floor{t} + 1}{2t} \right) + \frac{\floor{t} \cdot (\floor{t}^2 + \floor{t} + 2)}{n t}.
\end{align*}
Choosing a sufficiently large~$n$ so that $\frac{\floor{t} \cdot (\floor{t}^2 + \floor{t} + 2)}{n t} \leq \varepsilon$ completes the proof for the case $t\ge 2$.

\paragraph{Case~2: $1 \leq t < 2$.}
In this case, we have $\floor{t} = 1$ and thus $\alpha = \frac{\floor{t}^2 + \floor{t} + 2}{2} = 2$.
Recalling again that $n / \alpha$ is an integer, we partition the $n$ agents into three pairwise disjoint sets as follows:
\begin{align*}
&N_0 = \left\{ 1, \dots, \frac{n}{\alpha} - 1 \right\}, \\
&N_{\floor{t}} = \left\{ \frac{n}{\alpha}, \dots, \ceiling*{t \cdot \frac{n}{\alpha}} \right\}, \\
&D_{\floor{t}},
\end{align*}
where $N^* \coloneqq N_0 \cup N_{\floor{t}} = \left\{ 1, 2, \dots, \ceiling*{t \cdot \frac{n}{\alpha}} \right\}$ is our target $t$-cohesive group and $D_{\floor{t}}$ contains the remaining $n - \ceiling*{t \cdot \frac{n}{\alpha}}$ agents, who are ``dummy agents''.
From \eqref{eq:squeeze}, there is at least one agent in~$D_{\floor{t}}$:
\begin{align*}
n - \ceiling*{t \cdot \frac{n}{\alpha}} &= 2 \cdot \frac{n}{\alpha} - \ceiling*{t \cdot \frac{n}{\alpha}}
= (\floor{t} + 1) \cdot \frac{n}{\alpha} - \ceiling*{t \cdot \frac{n}{\alpha}} \geq 1.
\end{align*}
The set of indivisible goods is $\{g_1, g_2, \dots, g_{\ceiling{t}}\} \cup D_{\floor{t}}^G$, where $D_{\floor{t}}^G$ contains a single good.
The agents' preferences are as follows.
\begin{itemize}
\item The agents in~$N^*$ approve the goods in~$\{g_1, g_2, \dots, g_{\ceiling{t}}\}$.
\item The agents in~$N_{\floor{t}}\cup D_{\floor{t}}$ approve the good in~$D_{\floor{t}}^G$.
\end{itemize}
This completes the description of our instance.
Since the total number of agents is $n = 2\cdot n/\alpha$ but no good is approved by all $n$ agents, there is no $2$-cohesive group.
On the other hand, since the agents in~$N_{\floor{t}} \cup D_{\floor{t}}$ commonly approve the good in~$D_{\floor{t}}^G$ and
\[
\left| N_{\floor{t}} \cup D_{\floor{t}} \right| = n - \left( \frac{n}{\alpha} - 1 \right) = \frac{n}{2} + 1 > \frac{n}{2} = \frac{n}{\alpha},
\]
$N_{\floor{t}} \cup D_{\floor{t}}$ is $1$-cohesive.
Thus, GreedyEJR-M can start by identifying the $1$-cohesive group~$N_{\floor{t}} \cup D_{\floor{t}}$ and adding $D_{\floor{t}}^G$ to the allocation~$A$.
Once the agents in~$N_{\floor{t}} \cup D_{\floor{t}}$ are removed by GreedyEJR-M, the remaining agents in~$N_0$ do not form a $1$-cohesive group, so no more good is added to~$A$.

Note that $N^*$ is a $1$-cohesive group, since it has size $\ceiling*{t \cdot \frac{n}{\alpha}} \ge \frac{n}{\alpha}$ and the agents in it commonly approve $g_1$.
The average satisfaction of~$N^*$ with respect to the returned allocation~$A$ is
\begin{align*}
\frac{\sum_{i \in N^*} u_i(A)}{|N^*|}
&= \frac{|N_0| \cdot 0}{|N^*|} + \frac{\left| N_{\floor{t}} \right| \cdot \floor{t}}{|N^*|} \\
&= \frac{\left( \ceiling*{t \cdot \frac{n}{\alpha}} - \floor{t} \cdot \frac{n}{\alpha} + 1 \right) \cdot \floor{t}}{\ceiling*{t \cdot \frac{n}{\alpha}}} \\
&= \frac{\left( \ceiling*{\floor{t} \cdot \frac{n}{\alpha} + (t - \floor{t}) \cdot \frac{n}{\alpha}} - \floor{t} \cdot \frac{n}{\alpha} + 1 \right) \cdot \floor{t}}{\ceiling*{t \cdot \frac{n}{\alpha}}} \\
&\leq \frac{\left( \floor{t} \cdot \frac{n}{\alpha} + \ceiling*{(t - \floor{t}) \cdot \frac{n}{\alpha}} - \floor{t} \cdot \frac{n}{\alpha} + 1 \right) \cdot \floor{t}}{t \cdot \frac{n}{\alpha}} \\
&\leq \frac{\left( (t - \floor{t}) \cdot \frac{n}{\alpha} + 2 \right) \cdot \floor{t}}{t \cdot \frac{n}{\alpha}} \\
&= \frac{\floor{t} \cdot (t - \floor{t})}{t} + \frac{2 \alpha \floor{t}}{n t} \\
&= \floor{t} \cdot \left( 1 - \frac{\floor{t} + 1}{2t} \right) + \frac{2 \alpha \floor{t}}{n t},
\end{align*}
where the last equation follows from $\floor{t} = 1$.
Choosing a sufficiently large~$n$ so that $\frac{2 \alpha \floor{t}}{n t} \leq \varepsilon$ completes the proof for the case $t \in [1, 2)$, and therefore the proof of \Cref{thm:prop-degree-GreedyEJR-M}.
\end{proof}

In the indivisible-goods setting, \citet[Prop.~A.10]{LacknerSk22} showed that for positive integers $t$, the proportionality degree of MES is between $\frac{t-1}{2}$ and $\frac{t+1}{2}$.
Since Generalized MES satisfies EJR-1, \Cref{thm:EJR-1-LB-average-satisfaction} implies a lower bound of $\frac{t-2+1/t}{2}$ on its proportionality degree.
On the other hand, since any $t'$-cohesive group is also $t$-cohesive for $t\le t'$, Lackner and Skowron's result implies an upper bound of $\frac{\ceiling{t}+1}{2}$ for Generalized MES.

Finally, we prove that Generalized PAV has a significantly higher proportionality degree than the other two rules that we study.
In doing so, we extend a result of \citet{AzizElHu18} from the indivisible-goods setting.

\begin{theorem}
\label{thm:prop-degree-GPAV}
For any real number $t\ge 1$, the average satisfaction of a $t$-cohesive group with respect to a Generalized PAV allocation is greater than $t-1$.
\end{theorem}

\begin{proof}
The proof is almost identical to that of \Cref{thm:GPAV}.
Let $R'$ be a Generalized PAV allocation, and assume for contradiction that a $t$-cohesive group $N'$ has $\frac{1}{\lvert N'\rvert}\cdot \sum_{i \in N'} u_i(R') \le t-1$.
Since the agents in~$N'$ commonly approve a resource of size at least $t$, if the subset of this resource included in $R'$ has size larger than $t-1$, we would have $\frac{1}{\lvert N'\rvert}\cdot \sum_{i \in N'} u_i(R') > t-1$, a contradiction.
Hence, there exists a resource of size~$1$ approved by all agents in $N'$ but not included in $R'$.
The rest of the proof proceeds in the same way as that of \Cref{thm:GPAV}; in particular, the initial chain of inequalities starting with $H(R'')-H(R')$ still holds because $\sum_{i\in N'}u_i(R') \le |N'|\cdot(t-1)$.
\end{proof}

We also demonstrate in \Cref{app:tightness} that the bound $t-1$ is almost tight.

\section{Conclusion}

In this work, we have initiated the study of approval-based voting with mixed divisible and indivisible goods, which allows us to unify both the well-studied setting of multiwinner voting and the recently introduced setting of cake sharing.
We generalized three important rules from multiwinner voting to our setting, determined their relations to our proposed extensions of the EJR axiom, and investigated their proportionality degree.
In particular, we found that each of the three rules is superior in a certain way: \mbox{GreedyEJR-M} satisfies EJR-M, Generalized MES satisfies strong EJR-1 and can be computed in polynomial time, and Generalized PAV has a high proportionality degree.
Since we do not know whether GreedyEJR-M can be implemented in polynomial time, an intriguing open question is whether there exists a polynomial-time algorithm for computing an EJR-M allocation.
Further directions for future work include exploring other axioms such as \emph{proportional justified representation (PJR)} \citep{SanchezFernandezElLa17} and the \emph{core}.
In fact, one can define \mbox{PJR-M} and \mbox{PJR-1} analogously to our \mbox{EJR-M} and \mbox{EJR-1}---since these PJR axioms are weaker than the respective EJR axioms, our positive results on the EJR variants directly carry over to the PJR variants.
In addition, it could be interesting to generalize our results to a \emph{participatory budgeting} setting, where different parts of the resource may have different \emph{costs}.
This may entail, for instance, extending the EJR notions of \cite{PetersPiSk21} for discrete goods to accommodate mixed goods.

\section*{Acknowledgments}

This work was supported by ARC Laureate Project FL200100204 on ``Trustworthy AI'', by the Singapore Ministry of Education under grant number MOE-T2EP20221-0001, by the Deutsche Forschungsgemeinschaft under
grant BR 4744/2-1 and the Graduiertenkolleg ``Facets of
Complexity'' (GRK 2434),  and by an NUS Start-up Grant.
We would like to thank the anonymous reviewers for their valuable comments.

\bibliographystyle{plainnat}
\bibliography{aaai23}

\appendix

\section{Proof of \Cref{prop:EJR-M-strong-EJR-1-no-logic-relation}}
\label{app:no-logic-relation}

We prove that neither \EJRM nor \strongEJROne implies each other.
First, the allocation in \Cref{ex:Gen-MES} satisfies \strongEJROne but not \EJRM, which means that \strongEJROne does not imply \EJRM.

Next, we show that GreedyEJR-M does not satisfy \strongEJROne; it immediately follows that \EJRM does not imply \strongEJROne.
Consider an instance with $n = 5$ agents, four indivisible goods $g_1, \dots, g_4$, a cake of length~$1$, and $\alpha = 3$.
The agents' valuations over the resources are as specified in the following figure.
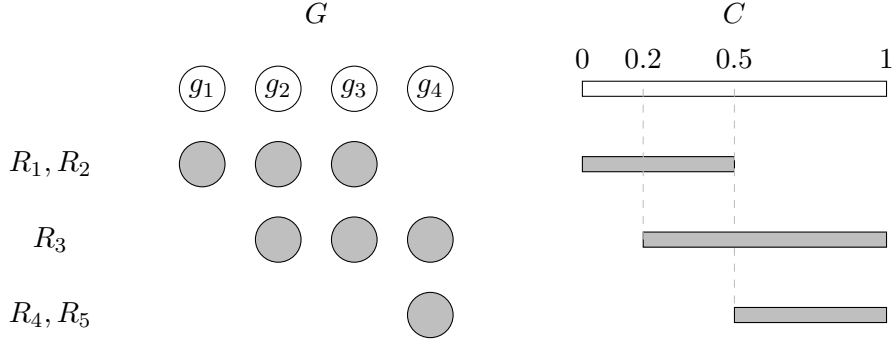
\begin{center}
\tikzstyle{vertex}=[circle,draw,minimum size=0.6cm,inner sep=0pt]
\tikzstyle{approvedVertex}=[circle,draw,minimum size=0.6cm,inner sep=0pt,fill=gray!50]
\tikzstyle{approvedPiece}=[fill=gray!50]
\begin{tikzpicture}
\node at (2.5,1) {$G$};
\node at (8,1) {$C$};
\draw (1,0) node[vertex] {$g_1$}
(2,0) node[vertex] {$g_2$}
(3,0) node[vertex] {$g_3$}
(4,0) node[vertex] {$g_4$};
\draw (6,-.1) rectangle ++(4,.2);
\node[label=above:{$0$}] at (6,0) {};
\node[label=above:{$0.2$}] at (6.8,0) {};
\node[label=above:{$0.5$}] at (8,0) {};
\node[label=above:{$1$}] at (10,0) {};

\node at (-1,-1) {$R_1, R_2$};
\draw (1,-1) node[approvedVertex] {}
(2,-1) node[approvedVertex] {}
(3,-1) node[approvedVertex] {};
\draw[approvedPiece] (6,-1.1) rectangle ++(2,.2);
\draw[dashed,color=gray!50] (8,0) to (8,-3);

\node at (-1,-2) {$R_3$};
\draw (2,-2) node[approvedVertex] {}
(3,-2) node[approvedVertex] {}
(4,-2) node[approvedVertex] {};
\draw[approvedPiece] (6.8,-2.1) rectangle ++(3.2,.2);
\draw[dashed,color=gray!50] (6.8,0) to (6.8,-2);

\node at (-1,-3) {$R_4, R_5$};
\draw (4,-3) node[approvedVertex] {};
\draw[approvedPiece] (8,-3.1) rectangle ++(2,.2);
\end{tikzpicture}
\end{center}

We start by stating several facts in relation to cohesive groups:
\begin{itemize}
\item First, agents~$1, 2, 3$ form a $1.8$-cohesive group, because $|\{1, 2, 3\}| = 3 = 1.8 \cdot \frac{5}{3}$ and the agents commonly approve a resource of size at least~$1.8$, that is, 
\[
s(\cap_{i \in \{1, 2, 3\}} R_i) = |\{g_2, g_3\}| + \ell([0.2, 0.5]) = 2.3 \geq 1.8.
\]
However, it is worth noting that the EJR-M condition on this group does \emph{not} work for any $t > 1.3$, as the agents in this group do not commonly approve a resource of size \emph{exactly} $t$ for any $t\in (1.3,1.8]$.

\item Similarly, agents~$3, 4, 5$ form a $1.5$-cohesive group, because $|\{3,4,5\}| = 3 \geq 1.5 \cdot \frac{5}{3}$ and $s(\bigcap_{i \in \{3, 4, 5\}} R_i) = |\{g_4\}| + \ell([0.5, 1]) = 1.5$.
In addition, the EJR-M condition with $t = 1.5$ works for this group.

\item Finally, agents~$1, 2$ form a $1.2$-cohesive group, because $|\{1,2\}| = 2 = 1.2 \cdot \frac{5}{3}$ and the agents commonly approve a resource of size at least~$1.2$.
Moreover, the EJR-M condition with $t = 1.2$ works for this group.
\end{itemize}

Taking as input the above instance, GreedyEJR-M first chooses $t^* = 1.5$, along with $N^* = \{3, 4, 5\}$ and $R^* = \{g_4\} \cup [0.5, 1]$.
We are left with $N' = \{1, 2\}$, and the next iteration of GreedyEJR-M chooses $t^* = 1.2$, $N^* = \{1, 2\}$, and, for example, $R^* = \{g_1\} \cup [0, 0.2]$.
Hence, GreedyEJR-M terminates with the allocation $A = \{g_1, g_4\} \cup [0, 0.2] \cup [0.5, 1]$.

Recall that agents~$1, 2, 3$ form a $1.8$-cohesive group, and note that $u_1(A) = u_2(A) = 1.2 < 1.8$ and $u_3(A) = 1.5 < 1.8$.
In particular, all agents in this $1.8$-cohesive group receive utility strictly less than~$1.8$.
Moreover, there is still a cake commonly approved by this group that is not included in $A$ (i.e., $[0.2, 0.5]$).
Hence, GreedyEJR-M does not satisfy \strongEJROne.
As a result, \EJRM does not imply \strongEJROne.

\section{(Almost) Tightness of \Cref{thm:prop-degree-GPAV}}
\label{app:tightness}

In their work, \citet{AzizElHu18} showed that, for indivisible-goods instances, the bound $t-1$ on the proportionality degree of PAV is tight for every positive integer $t$.
In fact, they proved a stronger statement that for any integer $t \ge 1$ and real number $\varepsilon > 0$, there exists an instance such that no allocation provides an average satisfaction of at least $t - 1 + \varepsilon$ to every $t$-cohesive group.
We extend their result to our setting by showing that for real numbers $t \ge 1$, the bound $t-1$ is essentially tight for large~$t$; our formal statement can be found below.

\begin{theorem}\label{thm:AS-upper-bound}
For any real number $t \geq 1$ and $\varepsilon >0$, let $z = t - \floor{t}$ denote the fractional part of~$t$.
There exists an indivisible-goods instance in which no allocation provides an average satisfaction of at least $t-1+\frac{z(1-z)}{t} + \varepsilon$ to all $t$-cohesive groups.
\end{theorem}

\begin{proof}
  We first assume that $t$ (and therefore $z$) is a rational number.
  Let $\gamma$ be a rational number in $(0, \varepsilon)$ such that $\gamma < 1-z$.
  Let $z = p/q$ and $\gamma = p'/q$ for some $p', p, q \in \mathbb{Z}$, and let $k = \floor{t}$ (so $t = k+z$).
  Our instance is given as follows.
  \begin{itemize}
      \item The set of agents $N$ can be partitioned into the following pairwise disjoint sets:
          \begin{itemize}
              \item $N_0$ contains $q\cdot ((k+1)z+k\gamma)$ agents;
              \item Each of $N_1, \ldots, N_{k+1}$ contains $q\cdot (1-z-\gamma)$ agents.
          \end{itemize}
          Therefore, the total number of agents is
          \begin{align*}
          n &= q\cdot ((k+1)z+k\gamma + (k+1)(1-z-\gamma))
          = q\cdot (k+1-\gamma).
          \end{align*}
      \item The resource consists of $(k+1)^2$ indivisible goods, which can be partitioned into $k+1$ disjoint sets $G_1, \ldots, G_{k+1}$, each with $k+1$ goods.
      \item The agents' preferences are such that for each $i \in \{1, \ldots, k+1\}$, the agents in $M_i \coloneqq N\setminus N_i$ commonly approve the goods in $G_i$.
      Note that each set $M_i$ consists of $q\cdot((k+1)z+k\gamma + k(1-z-\gamma)) = q\cdot (k+z) = q\cdot t$ agents.
      \item We set $\alpha = k+1-\gamma$.
      Since each set $M_i$ consists of $q\cdot t = t\cdot n/\alpha$ agents and these agents commonly approve $|G_i| = k+1 \ge t$ goods, $M_i$ is a $t$-cohesive group.
  \end{itemize}

  Because $\alpha = k+1-\gamma < k+1$ and our instance consists only of indivisible goods, any feasible allocation~$A$ can include at most $k$ goods.
  Since we have $k+1$ sets of goods $G_1, \ldots, G_{k+1}$, at least one of these sets is entirely excluded from~$A$.
  Consider an arbitrary allocation~$A$, and assume without loss of generality that $G_1$ is such a set.

  We now analyze the average satisfaction of the $t$-cohesive group~$M_1$ with respect to~$A$.
  Note that for each $i \in \{2,\dots,k+1\}$, any good selected from $G_i$ contributes a total utility of
  \begin{align*}
  |M_1 \cap M_i|
  &= |N \setminus (N_1 \cup N_i)| \\
  &= q\cdot (k+1-\gamma - 2(1-z-\gamma)) \\
  &= q\cdot(t-(1-z)+\gamma)
  \end{align*}
  to agents in $M_1$. Therefore, all goods in $A$ combined give a total utility of at most
  \[q\cdot k \cdot(t-(1-z)+\gamma) = q[(t-z)(t-(1-z))+k\gamma]\]
  to all agents in $M_1$.
  It follows that the average satisfaction of $M_1$ with respect to~$A$ is at most
  \begin{align*}
  \frac{q[(t-z)(t-(1-z))+k\gamma]}{qt}
  &= \frac{(t-z)(t-(1-z))+k\gamma}{t} \\
  &= \frac{(t^2-zt) - (t-z)(1-z)+k\gamma}{t}  \\
  &= \frac{t^2-zt + zt - t + z(1-z)+k\gamma}{t}\\
  &\leq t-1+\frac{z(1-z)}{t}+\gamma \\
  &< t-1+\frac{z(1-z)}{t}+\varepsilon.
\end{align*}
This completes the proof for the case where $t$ is a rational number.

Finally, assume that $t$ is irrational.
Let $t' > t$ and $\varepsilon' < \varepsilon$ be rational numbers such that $\floor{t'} = \floor{t}$ and $t'+ \frac{z'(1-z')}{t'}+ \varepsilon' < t+\frac{z(1-z)}{t}+ \varepsilon$, where $z' = t' - \floor{t'}$; the existence of such a pair $(t',\varepsilon')$ is guaranteed by the fact that, as we increase $t$ slightly, the value of $t + \frac{z(1-z)}{t}$ changes continuously.
From our previous argument, there exists an instance in which no allocation provides an average satisfaction of at least $t'-1 + \frac{z'(1-z')}{t'} + \varepsilon'$ to all $t'$-cohesive groups.
Suppose for contradiction that there is an allocation that provides an average satisfaction of at least $t-1 + \frac{z(1-z)}{t} + \varepsilon$ to all $t$-cohesive groups in this instance.
Because any $t'$-cohesive group is also $t$-cohesive, this allocation would also provide an average satisfaction of at least $t-1 + \frac{z(1-z)}{t} + \varepsilon > t'-1+\frac{z'(1-z')}{t'}+\varepsilon'$ to all $t'$-cohesive groups, a contradiction.
\end{proof}

We also show that the bound proved in \Cref{thm:AS-upper-bound} is---perhaps surprisingly---tight.

\begin{theorem}\label{thm:AS-lower-bound}
For any real number $t \geq 1$, let $z = t - \floor{t}$ denote the fractional part of~$t$.
Given any instance, there exists an allocation that provides an average satisfaction of at least $t - 1 + \frac{z (1 - z)}{t}$ to all $t$-cohesive groups.
\end{theorem}

\begin{proof}
Let $\mathcal{I} = \langle N, R, (R_i)_{i \in N}, \alpha \rangle$ be an instance, and let $k = \floor{t}$ (so $t = k + z$).
Let $\mathcal{T} = \{T_1, T_2, \dots, T_p\}$ be the set of all $t$-cohesive groups in instance~$\mathcal{I}$.
By definition, for each~$i \in \{1,2,\dots,p\}$, it holds that $|T_i| \geq t \cdot n/\alpha$ and $s(\bigcap_{j \in T_i} R_j) \geq t$.

We create another instance~$\widehat{\mathcal{I}} = \langle N, \widehat{R}, (\widehat{R}_i)_{i \in N}, \alpha \rangle$ with the same set of agents~$N$ and a modified resource $\widehat{R} = \bigcup_{i=1}^p \widehat{R}^{T_i}$ such that
\begin{itemize}
\item for any pair of distinct~$i, j \in \{1, \dots, p\}$, $\widehat{R}^{T_i} \cap \widehat{R}^{T_j} = \emptyset$;

\item for each~$i \in \{1,\dots,p\}$, $\widehat{R}^{T_i}$ consists of~$k$ \emph{indivisible} goods that are commonly approved (only) by the agents in~$T_i$.
\end{itemize}
Put differently, the resource~$\widehat{R}^{T_i}$ in instance~$\widehat{\mathcal{I}}$ corresponds to the resource~$\bigcap_{j \in T_i} R_j$ in instance~$\mathcal{I}$, but has a weakly smaller size.
The idea here is that we are ``worsening'' the instance in terms of the average satisfaction for all $t$-cohesive groups.
More specifically, given an allocation~$\widehat{R}'$ for instance~$\widehat{\mathcal{I}}$, we can select an allocation~$R'$ for instance~$\mathcal{I}$ as follows: for each $\widehat{R}^{T_i}$, if $\ell \le k$ (indivisible) goods from $\widehat{R}^{T_i}$ are included in~$\widehat{R}'$, then we include a resource of size $\ell$ from~$\bigcap_{j \in T_i} R_j$ in $R'$.
(Even if some resource gets included in~$R'$ multiple times during this process, it is effectively only included once.)
Clearly, $s(R') \leq s(\widehat{R}') \leq \alpha$.
Moreover, it can be seen that for each~$i \in \{1,\dots,p\}$,
\[
\frac{1}{|T_i|} \cdot \sum_{j \in T_i} u_j(R') \geq \frac{1}{|T_i|} \cdot \sum_{j \in T_i} u_j(\widehat{R}').
\]
It is worth noting that in instance~$\widehat{\mathcal{I}}$, each agent group~$T_i$ may no longer be a $t$-cohesive group, because their set of commonly approved goods is now $\widehat{R}^{T_i}$, whose size is only $k = \floor{t}$.
Nevertheless, our goal is to find an allocation~$\widehat{R}'$ for instance~$\widehat{\mathcal{I}}$ such that for every~$i \in \{1,\dots,p\}$, the average satisfaction of~$T_i$ with respect to~$\widehat{R}'$ in instance~$\widehat{\mathcal{I}}$ is at least $t - 1 + \frac{z (1-z)}{t}$, even if $T_i$ is not a $t$-cohesive group in instance~$\widehat{\mathcal{I}}$.
As a result, as we argued above, there exists a corresponding allocation~$R'$ for instance~$\mathcal{I}$ that has the same or better average satisfaction for all $t$-cohesive groups in instance~$\mathcal{I}$.

We apply the classic PAV rule to find the allocation~$\widehat{R}'$ for instance~$\widehat{\mathcal{I}}$ (which is an indivisible-goods instance).
That is, we choose~$\widehat{R}'$ that maximizes $H(\widehat{R}') = \sum_{i \in N} H_{u_i(\widehat{R}')}$, where $H_x \coloneqq 1 + \frac12 + \dots + \frac1x$ is the $x$-th harmonic number.
The following analysis is similar to that of \Cref{thm:GPAV}, but more refined.
By adding an indivisible good approved by no agent to the resource~$\widehat{R}$ as well as~$\widehat{R}'$ if necessary, we may assume without loss of generality that $s(\widehat{R}) = \floor{\alpha}$.

To show that~$\widehat{R}'$ satisfies the target average satisfaction value for each $T_i \in \mathcal{T}$, we assume for contradiction that there exists $T_* \in \mathcal{T}$ such that
\begin{equation}\label{eq:contradiction-assumption}
\frac{1}{|T_*|} \cdot \sum_{i \in T_*} u_i(\widehat{R}') < t - 1 + \frac{z (1-z)}{t} = t - z - 1 + \frac{z (t - z + 1)}{t} = k - 1 + \frac{(k+1) z}{t}.
\end{equation}
Since $\frac{(k+1) z}{t} < \frac{k + z}{t} = 1$, we have
\begin{equation}\label{eq:range}
k - 1 \leq k - 1 + \frac{(k+1) z}{t} < k,
\end{equation}
which means that there exists a good~$g^* \in \widehat{R}^{T_*}$ that is not selected in~$\widehat{R}'$.
Let $\widehat{R}'' \coloneqq \widehat{R}' \cup \{g^*\}$; so $|\widehat{R}''| = \floor{\alpha} + 1$.
We have
\begin{align*}
H(\widehat{R}'') - H(\widehat{R}') \geq \sum_{i \in T_*} \left( H_{u_i(\widehat{R}') + 1} - H_{u_i(\widehat{R}')} \right) = \sum_{i \in T_*} \frac{1}{u_i(\widehat{R}') + 1}.
\end{align*}

In what follows, we provide a lower bound on $\sum_{i \in T_*} \frac{1}{u_i(\widehat{R}') + 1}$.
First, for each~$i \in T_*$, let $y_i = u_i(\widehat{R}')$.
We thus have $\sum_{i \in T_*} \frac{1}{u_i(\widehat{R}') + 1} = \sum_{i \in T_*} \frac{1}{y_i + 1}$.
If there exists a pair $i, j \in T_*$ with $y_i \geq y_j + 2$, let $y_i' = y_i - 1$ and $y_j' = y_j + 1$.
It is easy to verify that
\[
\frac{1}{y_i + 1} + \frac{1}{y_j + 1} > \frac{1}{y_i' + 1} + \frac{1}{y_j' + 1}.
\]
By replacing~$y_i$ (resp.,~$y_j$) with~$y_i'$ (resp.,~$y_j'$), we decrease $\sum_{b \in T_*} \frac{1}{y_b + 1}$.
Repeat this process until, for every pair $i, j \in T_*$, it holds that $|y_i - y_j| \leq 1$.
We now have
\[
\sum_{i \in T_*} \frac{1}{u_i(\widehat{R}') + 1} \geq \sum_{i \in T_*} \frac{1}{y_i + 1}.
\]
Note that all $y_i$'s are integers and the following equation still holds:
\begin{equation}\label{eq:transform}
\sum_{i \in T_*} u_i(\widehat{R}') = \sum_{i \in T_*} y_i.
\end{equation}
Together with \eqref{eq:contradiction-assumption} and \eqref{eq:range}, we have that $y_i \leq k-1$ for some~$i \in T_*$, and therefore $y_i \leq k$ for all~$i \in T_*$.

If $\sum_{i \in T_*} y_i \leq |T_*| \cdot (k - 1)$, then
\[
\sum_{i \in T_*} \frac{1}{y_i + 1} \geq \frac{|T_*|^2}{\sum_{i \in T_*} (y_i + 1)} = \frac{|T_*|^2}{\sum_{i \in T_*} y_i + |T_*|} \geq \frac{|T_*|^2}{|T_*| \cdot (k-1) + |T_*|} \geq \frac{|T_*|}{t} \geq \frac{n}{\alpha},
\]
where the first transition follows from the inequality of arithmetic and harmonic means and the second-to-last transition from the fact that $t\ge k$.

Else, $\sum_{i \in T_*} y_i > |T_*| \cdot (k-1)$.
This means that there exists $i \in T_*$ such that $y_i = k$.
Let $\widetilde{T}_* \coloneqq \{i \in T_* \mid y_i = k\}$; we have $|\widetilde{T}_*| > 0$.
As argued earlier, for all~$i \in T_* \setminus \widetilde{T}_*$, it holds that $y_i = k-1$.
Recall from \eqref{eq:contradiction-assumption} and \eqref{eq:transform} that
\begin{align*}
|\widetilde{T}_*| \cdot k + (|T_*| - |\widetilde{T}_*|) \cdot (k-1) 
&= \sum_{i \in \widetilde{T}_*} k + \sum_{i \in T_* \setminus \widetilde{T}_*} (k-1)\\
&=\sum_{i \in T_*} y_i\\
&= \sum_{i \in T_*} u_i(\widehat{R}')
< |T_*| \cdot \left( k - 1 + \frac{(k+1) z}{t} \right),
\end{align*}
which implies that $|\widetilde{T}_*| < |T_*| \cdot \frac{(k+1) z}{t}$.
As a result, we have
\begin{align*}
\sum_{i \in T_*} \frac{1}{y_i + 1} = \sum_{i \in \widetilde{T}_*} \frac{1}{k+1} + \sum_{i \in T_* \setminus \widetilde{T}_*} \frac{1}{k} &= \frac{|\widetilde{T}_*|}{k+1} + \frac{|T_*| - |\widetilde{T}_*|}{k} \\
&= \frac{(k+1) \cdot |T_*| - |\widetilde{T}_*|}{k (k+1)} \\
&> \frac{(k+1) \cdot |T_*| - |T_*| \cdot \frac{(k+1) z}{t}}{k (k+1)} \\
&= \frac{|T_*| \cdot \left( 1 - \frac{z}{t} \right)}{k} = \frac{|T_*|}{t} \geq \frac{n}{\alpha}.
\end{align*}
Therefore, in either case, adding~$g^*$ increases the PAV-score of~$\widehat{R}'$ by at least~$n/\alpha$.

Finally, the rest of the proof proceeds in the same way as that of \Cref{thm:GPAV} when arguing about the marginal contribution of an indivisible good in $R''$ (which corresponds to $\widehat{R}''$ in our proof here).
For each good~$g \in \widehat{R}$, denote by $N_g \subseteq N$ the set of agents who approve it.
For each~$g \in \widehat{R}''$, we have
\[
H(\widehat{R}'') - H(\widehat{R}'' \setminus \{g\}) = \sum_{i \in N_g} \left( H_{u_i(\widehat{R}'')} - H_{u_i(\widehat{R}'') - 1} \right) = \sum_{i \in N_g} \frac{1}{u_i(\widehat{R}'')}.
\]
Letting $N_+$ consist of the agents~$i\in N$ with $u_i(\widehat{R}'') > 0$, we get
\begin{align*}
\sum_{g \in \widehat{R}''} (H(\widehat{R}'') - H(\widehat{R}'' \setminus \{g\})) = \sum_{g \in \widehat{R}''} \sum_{i \in N_g} \frac{1}{u_i(\widehat{R}'')} = \sum_{i \in N_+} \sum_{g \in \widehat{R}'' \cap \widehat{R}_i} \frac{1}{u_i(\widehat{R}'')} = |N_+| \leq n.
\end{align*}
If there is a good~$g \in \widehat{R}''$ such that $H(\widehat{R}'') - H(\widehat{R}'' \setminus \{g\}) < n/\alpha$ (clearly, $g \neq g^*$), we can replace~$g$ with~$g^*$ in~$\widehat{R}'$ and obtain a higher PAV-score, contradicting the definition of~$\widehat{R}'$.
Hence, we may assume that $H(\widehat{R}'') - H(\widehat{R}'' \setminus \{g\}) \geq n/\alpha$ for every~$g \in \widehat{R}''$.
It follows that
\[
n \geq \sum_{g \in \widehat{R}''} (H(\widehat{R}'') - H(\widehat{R}'' \setminus \{g\})) \geq |\widehat{R}''| \cdot \frac{n}{\alpha}.
\]
Therefore, we have that $|\widehat{R}''| \leq \alpha$, contradicting the fact that $|\widehat{R}''| = \floor{\alpha} + 1$.
This completes the proof.
\end{proof}

The proof of \Cref{thm:AS-lower-bound}, however, does not show that there exists an allocation with the claimed average satisfaction guarantee for all~$t$ \emph{simultaneously}, as the creation of the new instance~$\widehat{\mathcal{I}}$ depends on~$t$.
While it is conceivable that Generalized PAV achieves this simultaneous guarantee, proving (or disproving) this seems to be a challenging task.

\end{document}